%% file: main.tex
\documentclass[12pt]{article}
\usepackage{amsmath,scalerel}
\usepackage{amsthm}
\usepackage{mathrsfs} 
\usepackage{amssymb}
\usepackage{graphicx}
\usepackage[shortlabels]{enumitem}
\usepackage{natbib}
\usepackage{url} 
\usepackage{accents}
\usepackage{booktabs}

\usepackage{tikz}
\usetikzlibrary{knots}
\usetikzlibrary{arrows.meta}
\usetikzlibrary{positioning}
\usetikzlibrary{shapes.geometric}
\usetikzlibrary{decorations.pathreplacing}
\usetikzlibrary{calc}
\usetikzlibrary{backgrounds}

\usepackage[many]{tcolorbox}

\usepackage{algorithm}
\usepackage{enumitem}

\usepackage[colorlinks=true,linkcolor=blue,citecolor=blue]{hyperref}

\newcommand{\blind}{1}

\newtheorem{definition}{Definition}[section]
\newtheorem{example}[definition]{Example}
\newtheorem{theorem}[definition]{Theorem}
\newtheorem{corollary}[definition]{Corollary}
\newtheorem{proposition}[definition]{Proposition}

\newcommand{\rmd}{\mathrm{d}}
\newcommand{\bseq}[2]{[#1\,{:}\,#2]}
\newcommand{\bseqo}[2]{[#1\,{:}\,#2)}
\newcommand{\fk}{\mathcal{M}}
\newcommand{\fkclass}{\mathscr{M}}
\newcommand{\knclass}{\mathscr{K}}

\newcommand{\kn}{\mathcal{K}}

\newcommand{\nn}[1]{\mathbb{N}_{#1}}
\newcommand{\simiid}{\overset{\text{iid}}{\sim}}

\begin{document}

\def\spacingset#1{\renewcommand{\baselinestretch}%
{#1}\small\normalsize} \spacingset{1}


\date{October 29, 2025}

\if1\blind
{
  \title{\bf Knots and variance ordering of sequential Monte Carlo algorithms}
  \author{Joshua J Bon\thanks{JJB was supported by the European Union under
the GA 101071601, through the 2023-2029 ERC Synergy grant OCEAN. AL was supported by the Engineering and Physical Sciences Research Council (EP/Y028783/1).}\hspace{.2cm}\\
    School of Mathematical Science, Adelaide University\\
    and \\
    Anthony Lee \\
     School of Mathematics, University of Bristol}
  \maketitle
} \fi

\if0\blind
{
  \bigskip
  \bigskip
  \bigskip
  \begin{center}
    {\LARGE\bf Knots and variance ordering of sequential Monte Carlo algorithms}
\end{center}
  \medskip
} \fi

\bigskip
\begin{abstract}
Sequential Monte Carlo algorithms, or particle filters, are widely used for approximating intractable integrals, particularly those arising in Bayesian inference and state-space models. We introduce a new variance reduction technique, the knot operator, which improves the efficiency of particle filters by incorporating potential function information into part, or all, of a transition kernel. The knot operator induces a partial ordering of Feynman--Kac models that implies an order on the asymptotic variance of particle filters, offering a new approach to algorithm design. We discuss connections to existing strategies for designing efficient particle filters, including model marginalisation. Our theory generalises such techniques and provides quantitative asymptotic variance ordering results.
We revisit the fully-adapted (auxiliary) particle filter using our theory of knots to show how a small modification guarantees an asymptotic variance ordering for all relevant test functions.  
\end{abstract}

\noindent%
{\it Keywords:} Particle filters, Feynman--Kac models, variance ordering, Rao--Blackwellisation
\vfill

\newpage 
\tableofcontents

\newpage
\spacingset{1.2} 
\section{Introduction}
\label{sec:intro}
Sequential Monte Carlo (SMC) algorithms, or particle filters, are malleable tools for estimating intractable integrals. These algorithms generate particle approximations for a sequence of probability measures on a path space, typically specified as a discrete-time Feynman--Kac model for which the normalising constant is intractable. Particle filters construct such approximations using Monte Carlo simulation, importance weighting, and resampling to propagate particles through this sequence of probability measures. 

SMC algorithms are used in diverse areas including signal processing \citep{gustafsson2002particle, doucet2005monte}, object tracking \citep{mihaylova2014overview,wang2017survey}, robotics \citep{thrun2002particle, stachniss2014particle}, econometrics \citep{lopes2011particle, creal2012survey, kantas2015particle}, weather forecasting \citep{fearnhead2018particle,van2019particle}, epidemiology \citep{endo2019introduction,temfack2024review}, and industrial prognostics \citep{jouin2016particle}. SMC algorithms are also used extensively in Bayesian posterior sampling \citep[i.e. SMC samplers,][]{del2006sequential,dai2022invitation} and in other difficult statistical tasks, such as rare event estimation \citep{cerou2012sequential}. Recently too, SMC algorithms have been used in areas of machine learning such as reinforcement learning \citep{lazaric2007reinforcement,lioutas2023critic,macfarlane2024spo} and denoising diffusion models \citep{cardoso2024monte,phillips2024particle} for example.

The extensive use of SMC algorithms across sciences and their ubiquity in computational statistics can be explained by the generality of their specification. SMC can be used in many different statistical problems and there are often several types of particle filters for a specific case. As SMC is fundamentally related to importance sampling, it is typical that several components of the algorithm can be altered, and accommodated for, by weighting without affecting the target probability measure. SMC samplers have further degrees of freedom as the path of distributions targeted can also be selected. Given this malleability, the design of efficient SMC algorithms remains an active area of research.

In the canonical SMC algorithm, the bootstrap particle filter \citep{gordon1993novel}, information is incorporated into the particle system according to the time of observation. Methods such as the auxiliary particle filter \citep{pitt1999filtering,johansen2008note}, look-ahead particle filters \citep{lin2013lookahead}, and model twisting methods \citep{guarniero2017iterated,heng2020controlled}, define new particle filters that incorporate varying degrees of future information into the current time step. Idealised versions of these algorithms reduce or eliminate variance in the particle system, producing more accurate particle approximations at the cost of increased computation. Typically, these idealised filters are not computationally tractable or are prohibitively expensive to calculate, so practical methods frequently employ approximations to an ideal filter.

In the simplest case, locally (i.e.\ conditional) optimal proposal distributions are used to define new particle filters where Monte Carlo simulation is adapted to the current potential information \citep{pitt1999filtering,doucet2000sequential}. The locally optimal proposal ensures the variance of the importance weights, conditional on a current particle state, is zero. The so-called fully-adapted (auxiliary) particle filter may reduce the overall variance in the particle system and has demonstrated good empirical performance. However, \citet{johansen2008note} note that such strategies do not guarantee an asymptotic variance reduction for a given test function.

When it is not possible to implement locally optimal proposals exactly, it may still be possible to find a proposal which reduces the variance of the importance weights.  Rao--Blackwellisation adapts a subset of the state space to current potential information, and has freedom in the choice of subset \citep{chen2000mixture,doucet2002particle,andrieu2002particle,schon2005marginalized}. Furthermore, heuristic approximations to these optimal proposals, or indeed Rao--Blackwellisation strategies, are also possible and often have good empirical performance \citep{doucet2000sequential}.

Extensions of adaptation using potential information beyond the current time have been explored in look-ahead methods \citep{lin2013lookahead} and typically employ approximations as the exact schemes are intractable. Recently, methods for iterative approximations to the optimal change of measure for the entire path space have seen interest. These rely on twisted Feynman--Kac models which generalise locally optimal proposals and look-ahead methods \citep{guarniero2017iterated,heng2020controlled}. In theory, applying a particle filter to an optimally twisted Feynman--Kac model results in a zero variance estimate of the model's normalising constant. Whilst in practice, iteratively estimating these models has been shown to dramatically reduce the variance for various test functions.

SMC samplers can also benefit from the aforementioned adaptation strategies but also have other degrees of freedom that are more prevalent in the literature on variance reduction \citep{del2006sequential}. Moreover, it is often more difficult to implement any exact adaptation strategies in SMC samplers as the weights take a more complex form, though \citet{dau2022waste} show this is possible in certain settings. Despite this, twisted Feynman--Kac models have been used successfully in SMC samplers \citep{heng2020controlled}. 

A major limitation in the study of optimal proposals, exact adaptation, and Rao--Blackwellisation is their theoretical underpinning. These methods are typically justified by appealing to minimisation of the conditional variance of the importance weights, or reducing variance of the joint weight over the entire path space. They do not give a theoretical guarantee on variance reductions for particle approximations of a test function, arising from a particle filter with resampling. Approximations and heuristics motivated by these methods are also subject to this unsatisfactory understanding. 
Further, methods using twisted Feynman--Kac models are only optimal for the normalising constant estimate of the model in the idealised version. Achieving this in practice is not guaranteed, though empirical performance can be strong at the cost of additional computation.

This paper contributes knots for Feynman--Kac models, a technique to reduce the asymptotic variance for particle approximations. We generalise and unify `full' adaptation and Rao--Blackwellisation, improving the theoretical underpinning for these methods, whilst providing a highly flexible strategy for designing new algorithms with guaranteed asymptotic variance reduction. Further, we resolve the discrepancy between the theoretical analysis of fully-adapted auxiliary particle filters and their demonstrated empirical performance. We demonstrate that a small change to particle filters with `full' adaptation can guarantee an asymptotic variance reduction, resolving the counterexample in \citet{johansen2008note}. 

Our techniques lead naturally to a partial ordering on general Feynman--Kac models whose order implies superiority in terms of the asymptotic variance of particle approximations --- a first for variance ordering of SMC algorithms by their underlying Feynman--Kac model. Further, we determine optimal knots and optimal sequences of knots to assist with SMC algorithm design.

The paper is structured as follows. We first review the background of SMC algorithms, including Feynman--Kac models and the asymptotic variance of particle approximations in Section~\ref{sec:background}. Knots are introduced in Section~\ref{sec:tyingknots} where we discuss their properties as well as invariances and equivalences of models before and after the application of a knot. Section~\ref{sec:var} contains our main variance reduction results for knots, whereas Section~\ref{sec:adaptedknots} discusses the optimality of knots. Section~\ref{sec:tyingterminalknots} provides variance and optimality results for terminal knots which require special treatment. Lastly, Section~\ref{sec:apps} contains examples including `full' adaptation and Rao--Blackwellisation as special cases of knots, and an illustrative example that is a hybrid between these two cases. Proofs of results are deferred to the appendix.

\section{Background}\label{sec:background}

Feynman--Kac models are path measures that can represent the evolution of particles generated by sequential Monte Carlo algorithms. A Feynman--Kac model can be constructed by weighting a Markov process. We consider algorithms for discrete-time models and hence restrict focus to a process specified by an initial distribution, a sequence of non-homogeneous Markov kernels, and potential functions for weights. Before describing these discrete-time Feynman--Kac models, we first introduce our notation.

\subsection{Notation}

Let $y_{i:j}$ be the vector $(y_i, y_{i+1},\ldots, y_j)$ when $i<j$ and $y_{k:k} = y_{k}$. For integers $n_0 < n_1$, we denote the set of integers as $\bseq{n_0}{n_1} = \{n_0,n_0+1,\ldots,n_1\}$ and write the set of natural numbers as $\nn0 = \{0,1,\ldots\}$ and $\nn1 = \{1,2,\ldots\}$. The function mapping any input to unit value is denoted by $1(\cdot)$ and the indicator function for a set $S$ is $1_S$. If $f$ and $g$ are functions, then $f\cdot g$ defines the map $x \mapsto f(x)g(x)$ whilst $f\otimes g = (x,y) \mapsto f(x)g(y)$. We denote the zero set for a function $f:\mathsf{X} \rightarrow \mathbb{R}$ as $S_0(f) = \{x \in \mathsf{X}:f(x)= 0\}$ .

Let $(\mathsf{X},\mathcal{X})$ be a measurable space. If $\mu$ is a measure on $(\mathsf{X},\mathcal{X})$ and the function ${\varphi:\mathsf{X} \rightarrow \mathbb{R}}$ then let $\mu(\varphi) = \int \varphi(x) \mu(\rmd x)$ and if $S \in \mathcal{X}$ then let $\mu(S) = \mu(1_S)$. We use $\mathcal{L}(\mu)$ to denote the class of functions that are integrable with respect to a measure $\mu$. 

In addition to $(\mathsf{X},\mathcal{X})$, let $(\mathsf{Y},\mathcal{Y})$ be a measurable space. When referring to a kernel we consider non-negative kernels, say $K:(\mathsf{X},\mathcal{Y}) \rightarrow [0,\infty)$. We define $K(\varphi)(\cdot) = \int K(\cdot,\rmd y) \varphi(y)$ for a function $\varphi: \mathsf{Y} \rightarrow \mathbb{R}$, $\mu K(\cdot) = \int \mu(\rmd x) K(x,\cdot)$ for a measure $\mu$ on $(\mathsf{X},\mathcal{X})$, and the tensor product as $(\mu \otimes K)(\rmd [x,y]) = \mu(\rmd x) K(x,\rmd y)$. The composition of two non-negative kernels, say~$L:(\mathsf{W},\mathcal{X}) \rightarrow [0,\infty)$ and $K$, is defined as $LK(w, \cdot) = \int L(w, \rmd x)K(x, \cdot)$ and is a right-associative operator, whilst the tensor product is $(L \otimes K)(w,\rmd [x,y]) = L(w,\rmd x) K(x,\rmd y)$. A Markov kernel is a non-negative kernel~$K$ such that $K(x,\mathsf{Y}) = 1$ for all $x\in\mathsf{X}$. We denote the identity kernel by~$\mathrm{Id}$ and the degenerate probability measure at $x$ as $\delta_x$.

We make use of twisted Markov kernels and will use the following superscript notation when describing these. If $K:(\mathsf{X},\mathcal{Y}) \rightarrow [0,1]$ is a Markov kernel and $H:\mathsf{Y} \rightarrow [0,\infty)$ is integrable w.r.t.\ $K(x,\cdot)$ for all $x\in\mathsf{X}$, let $K^H:(\mathsf{X},\mathcal{Y}) \rightarrow [0,1]$ be defined such that
\begin{equation*}
    K^H(x,\rmd y) = 
    \begin{cases}
        \frac{K(x,\rmd y)H(y)}{K(H)(x)} &~\text{if}~ K(H)(x)>0,\\
        K(x,\rmd y)  &~\text{if}~ K(H)(x)=0.
    \end{cases}
\end{equation*}
The choice to use $K$ when $K(H)(x)=0$ is somewhat arbitrary, but useful for formalising our mathematical arguments. Particle filter approximations do not depend on this choice. Similarly, if $\mu$ is a probability measure on $(\mathsf{Y},\mathcal{Y})$ then $\mu^H$ will be defined as $\mu^H(\rmd y) = \frac{\mu(\rmd y)H(y)}{\mu(H)}$ if~$\mu(H)>0$ and $\mu^H(\rmd y) = \mu(\rmd y)$ if~$\mu(H)=0$. If~$K$ is the identity kernel or degenerate probability measure, we take $K^H=K$ by convention.

The categorical distribution is denoted by $\mathcal{C}(q_{1},q_{2}\ldots,q_{m})$ defined with positive weights $(q_{1},q_{2}\ldots,q_{m})$ on support $\bseq{1}{m}$ with probabilities~$p_i = \left(\sum_{j=1}^{m}q_{j}\right)^{-1}q_{i}$ for $i \in \bseq{1}{m}$.

\subsection{Discrete-time Feynman--Kac models}\label{sec:fkbackground}

To define a Discrete-time Feynman--Kac model, we require a notion of compatible kernels, we refer to as composability. Composability is also used when we define knots in Section~\ref{sec:tyknot}.
\begin{definition}[Composability]
Let $\mathsf{Y}_p$ be a space and $(\mathsf{X}_p,\mathcal{X}_p)$ be a measurable space for $p\in \{1,2\}$. Let $M_1:(\mathsf{Y}_1, \mathcal{X}_1) \rightarrow [0,1]$ be a non-negative kernel (or measure, $M_1: \mathcal{X}_1 \rightarrow [0,1]$) and $M_2:(\mathsf{Y}_2, \mathcal{X}_2) \rightarrow [0,1]$ be a non-negative kernel. If $\mathsf{X}_1 \subseteq \mathsf{Y}_2$ then $M_1M_2$ is well-defined and we say that $M_1$ and $M_2$ are composable. 
\end{definition}

Consider measurable spaces $(\mathsf{X}_p,\mathcal{X}_p)$ for $p \in \bseq{0}{n}$, then let $M_0:\mathcal{X}_0\rightarrow [0,1]$ be a probability measure, $M_p: (\mathsf{X}_{p-1}, \mathcal{X}_p) \rightarrow [0,1]$ for $p \in \bseq{1}{n}$ be Markov kernels, and consider potential functions $G_p:\mathsf{X}_p\rightarrow [0,\infty)$ that are integrable with respect to $M_p(x,\cdot)$ for all $x\in\mathsf{X}_{p-1}$ and $p \in \bseq{0}{n}$.

\begin{definition}[Discrete-time Feynman--Kac model]\label{def:fkmodel}
A predictive Feynman--Kac model with horizon $n \in \nn{0}$ consists of an initial distribution $M_0$, Markov kernels $M_p$ for $p 
\in \bseq{1}{n}$, and potential functions $G_p$ for $p 
\in \bseq{0}{n-1}$ such that $M_{p-1}$ and $M_{p}$ are composable for $p \in \bseq{1}{n}$. In addition to the above, an updated Feynman--Kac model includes a potential $G_n$ at the terminal time.
\end{definition}
We will refer to a generic updated Feynman--Kac model with calligraphic notation $\fk$ or specifically the collection $(M_{0:n},G_{0:n})$. This specification includes both predictive and updated Feynman--Kac models, by taking $G_n = 1$ for the former. We will use $\fkclass_{n}$ to denote the class of discrete-time Feynman--Kac models with horizon $n$.

The initial measure, kernels, and potentials define a sequence of predictive measures starting with $\gamma_{0}=M_0$, and evolving by
\begin{equation}\label{eq:marginalrec}
    \gamma_{p+1} = \int \gamma_{p}(\rmd x_{p}) G_{p}(x_{p})M_{p+1}(x_{p}, \cdot)
\end{equation}
for $p \in \bseq{0}{n-1}$. The terminal measure can be thought of as the expectation over a path space, that is
\begin{equation}\label{eq:pathspaceexp}
    \gamma_n(\varphi) = \mathbb{E}\left[\varphi(X_n)\prod_{p=1}^{n-1} G_p(X_p) \right]
\end{equation}
with respect to a non-homogeneous Markov chain, specified by $X_0\sim M_0$ and $X_p\sim M_p(X_{p-1},\cdot)$ for $p \in \bseq{1}{n}$.

In comparison, the time-$p$ updated measures use potential information at time $p$ and are defined by $\hat\gamma_p(\rmd x_p) = \gamma_p(\rmd x_p) G_p(x_p)$ for $p \in \bseq{0}{n}$. The predictive and updated measures have normalised counterparts
\begin{equation*}
    \eta_p(\rmd x_p) = \frac{\gamma_p(\rmd x_p)}{\gamma_p(1)}, \quad \hat\eta_p(\rmd x_p) = \frac{\hat\gamma_p(\rmd x_p)}{\hat\gamma_p(1)},
\end{equation*}
which are probability measures for $p \in \bseq{0}{n}$. The path space representation of the updated terminal measure can be expressed by considering $\hat\gamma_n(\varphi) = \gamma_n(G_n\cdot \varphi)$ using \eqref{eq:pathspaceexp}.  

Lastly, an important quantity for the asymptotic variance calculation are the $Q_{p,n}$ kernels are defined as follows. Consider kernels $Q_{p}(x_{p-1}, \rmd x_p) = G_{p-1}(x_{p-1})M_p(x_{p-1}, \rmd x_p)$ for $p \in \bseq{1}{n}$ then let $Q_{n,n} = \mathrm{Id}$, $Q_{n-1,n} = Q_{n}$, and continue with $Q_{p,n} = Q_{p+1}\cdots Q_n = Q_{p+1}Q_{p+1,n}$ for $p \in \bseq{0}{n-2}$. In contrast to the expectation presented in \eqref{eq:pathspaceexp}, the $Q_{p,n}$ kernels are conditional expectations on the same path space, that is for $p \in \bseq{0}{n-1}$
\begin{equation*}
   Q_{p,n}(\varphi)(x_p) = \mathbb{E}\left[ \varphi(X_n) \prod_{t=p}^n G_t(X_t) \;\bigg{\vert}\; X_p=x_p \right].
\end{equation*}
In other words, at time $p$, the $Q_{p,n}$ complete the model in the sense that $\gamma_pQ_{p,n} = \gamma_n$.

\subsection{SMC and particle filters}

Sequential Monte Carlo algorithms, and in particular particle filters, approximate Feynman--Kac models by iteratively generating collections of points, denoted by $\zeta_{p}^{i}$ for~$i \in \bseq{1}{N}$, to approximate the sequence of probability measures~$\eta_{p}$ for $p \in \bseq{0}{n}$. We consider the bootstrap particle filter \citep{gordon1993novel} to approximate the terminal measure $\eta_n$ or its updated counterpart, which we simply refer to as a particle filter and describe in Algorithm~\ref{alg:bpf}. Different particle filters can be achieved by varying the underlying Feynman--Kac model whilst preserving the targets of the particle approximations of interest.
\begin{algorithm}
\caption{A Particle Filter}
\label{alg:bpf}
\textbf{Input:} Feynman--Kac model $\fk = (M_{0:n},G_{0:n})$ \begin{enumerate}[topsep=0pt,itemsep=-1ex,partopsep=0ex,parsep=0ex]
    \item Sample initial $\zeta^{i}_{0} \simiid M_{0}$ for $i \in \bseq{1}{N}$
    \item For each time $p \in [1:n]$
    \begin{enumerate}[label=\alph*.,topsep=0pt,itemsep=-1ex,partopsep=0ex,parsep=1ex]
        \item Sample ancestors $A^{i}_{p-1} \sim \mathcal{C}\left(G_{p-1}(\zeta^{1}_{p-1}), \ldots, G_{p-1}(\zeta^{N}_{p-1}) \right)$ for $i \in \bseq{1}{N}$
        \item Sample prediction $\zeta^{i}_{p} \sim M_{p}(\zeta^{A^{i}_{p-1}}_{p-1},\cdot)$ for $i \in \bseq{1}{N}$
    \end{enumerate}
\end{enumerate}
\textbf{Output:} Terminal particles $\zeta^{1:N}_n$
\end{algorithm}
After running a particle filter the approximate terminal predictive measures are
\begin{equation*}
    \eta_{n}^{N} = \frac{1}{N}\sum_{i=1}^{N}\delta_{\zeta^{i}_n}, \quad \gamma_{n}^{N} = \left\{\prod_{t=0}^{n-1} \eta_t^{N}(G_t)\right\}\eta_{n}^{N}.
\end{equation*}
Similarly, the updated terminal measures are
\begin{align*}
    \hat\eta_{n}^{N} = \sum_{i=1}^{N}W_{n}^{i}\delta_{\zeta^{i}_p}, \quad \hat\gamma_{n}^{N} = \left\{\prod_{t=0}^{n} \eta^{N}_t(G_t)\right\}\hat\eta_{n}^{N},
\end{align*}
where $W_{n}^{i} =  \frac{G_n(\zeta^{i}_n)}{\sum_{j=1}^N G_n(\zeta^{j}_n)}$.

\subsection{Asymptotic variance of particle approximations}\label{sec:asyvar}

The canonical asymptotic variance map~${\sigma^2: \mathcal{L}(\eta_n) \rightarrow [0,\infty]}$ is defined as
\begin{equation}
    \label{eq:asyvarpred}
    \sigma^2(\varphi) = \sum_{p=0}^n v_{p,n}(\varphi),\quad v_{p,n}(\varphi) = \frac{\gamma_p(1)\gamma_p(Q_{p,n}(\varphi)^2)}{\gamma_n(1)^2} - \eta_n(\varphi)^2,
\end{equation}
for a particle filter following Algorithm~\ref{alg:bpf}.
Under various conditions, particle approximations relate to this variance by way of Central Limit Theorems \citep[CLT,][]{del2004feynman}. 
CLTs establish that the terminal predictive measure,  ${\gamma_n^N(\varphi) \rightarrow \gamma_n(\varphi)}$ in probability as $N\rightarrow \infty$, and
\begin{equation*}
\sqrt{N}\left(\frac{\gamma_n^N(\varphi) - \gamma_n(\varphi)}{\gamma_n(1)}\right) \rightarrow \mathcal{N}\left(0, \sigma^2(\varphi) \right),
\end{equation*}
in distribution. Whilst for the terminal predictive probability measure, ${\eta_n^N(\varphi) \rightarrow \eta_n(\varphi)}$
in probability as $N\rightarrow \infty$, and
$\sqrt{N}\left(\eta_n^N(\varphi) - \eta_n(\varphi)\right) \rightarrow \mathcal{N}\left(0, \sigma^2(\varphi - \eta_n(\varphi)) \right)$ in distribution. It is also possible to establish convergence in moments \citep{lee2018variance}. In this case, ${\gamma_n^N(\varphi) \rightarrow \gamma_n(\varphi)}$ and ${\eta_n^N(\varphi) \rightarrow \eta_n(\varphi)}$ almost surely, whilst 
\begin{equation*}
    N \,\mathrm{Var}\left\{\gamma_n^N(\varphi)/\gamma_n(1)\right\} \rightarrow \sigma^2(\varphi), \quad
    N \,\mathbb{E}\left[\left\{\eta_n^N(\varphi)-\eta_n(\varphi)\right\}^2 \right]\rightarrow \sigma^2(\varphi - \eta_n(\varphi)).
\end{equation*}
When $\sigma^2(\varphi)$ is finite, CLTs and convergence statements motivate the use of the asymptotic variance map to characterise the theoretical performance of a particle filter. 

We analyse the asymptotic variance map directly, assuming that $\sigma^2(\varphi)$ is finite, without additional assumptions. As such, for the predictive measure of a given Feynman--Kac model~$\fk$, we will consider functions $\varphi \in \mathcal{F}(\fk)$ such that
$\mathcal{F}(\fk) = \{\varphi \in \mathcal{L}(\eta_n): \sigma^2(\varphi)<\infty\}$.
To interpret our results, we will require a CLT to hold for the model $\fk$ in question. For general Feynman--Kac models, for example, bounded potential functions and a bounded test function would suffice, but this need not be the case.

When applicable, analogous CLTs also hold for the updated marginal (probability) measures by using $\hat\sigma^2(\varphi)$ and $\hat\sigma^2(\varphi - \hat\eta_n(\varphi))$ in place of $\sigma^2(\varphi)$ and $\sigma^2(\varphi - \eta_n(\varphi))$ respectively, where 
\begin{equation}
    \label{eq:asyvarupdate}
    \hat\sigma^2(\varphi) = \sum_{p=0}^n \hat v_{p,n}(\varphi),\quad \hat v_{p,n}(\varphi) = \frac{v_{p,n}(G_n \cdot \varphi)}{\eta_n(G_n)^2}.
\end{equation}
For updated measures, our analysis then considers the class of test functions
$\hat{\mathcal{F}}(\fk) = \{\varphi \in \mathcal{L}(\hat\eta_n): \hat\sigma^2(\varphi)<\infty\}$. In our discussions, we will refer to functions in the classes~$\mathcal{F}(\fk)$ and~$\hat{\mathcal{F}}(\fk)$ as relevant test functions.

\section{Tying knots in Feynman--Kac models}\label{sec:tyingknots}

In order to present our procedure for reducing the variance of SMC algorithms, we  must define how the underlying Feynman--Kac model is modified. To this end we define a \textit{knot}, encoding the details of the modification, and the \textit{knot operator} which describes how a knot is applied to a Feynman--Kac model.

A knot is specified by a time, $t$, and composable Markov kernels, $R$ and $K$, which can be used to modify suitable Feynman--Kac models whilst preserving the terminal measure. In one view, a $t$-knot modifies a Feynman--Kac model by partially adapting the Markov kernel $M_t$ to potential information at time $t$, though repeated applications of knots allow adaptation to potential information beyond just the next time.

We will introduce knots formally in Section~\ref{sec:tyknot} and a convenient notion for their simultaneous application, \textit{knotsets}, in Section~\ref{sec:tyknotset}.
Knots at time~$n$ are considered later, in Section~\ref{sec:tyingterminalknots}, as \textit{terminal knot}s require special treatment and have a smaller scope compared to regular knots (time $t<n$).

\subsection{Knots}\label{sec:tyknot}

We begin with a formal definition of a knot and what it means for a knot to be compatible with a Feynman--Kac model.
\begin{definition}[Knot] \label{def:knot}
A knot is a triple $\kn=(t,R,K)$, consisting of a time $t \in \nn{0}$, and an ordered pair of composable Markov kernels $R$ and $K$. Note that when $t=0$, $R$ is a probability distribution.
\end{definition}
For compactness we will often refer to knots abstractly as $\kn$, meaning $\kn = (t,R,K)$ for some $t$, $R$ and $K$, and when emphasis on the time component of the knot is required, we will refer to $\kn$ a $t$-knot.
\begin{definition}[Knot compatibility] \label{def:kncompat}
For $t<n$, a knot $\kn=(t,R,K)$ is compatible with a Feynman--Kac model $\fk=(M_{0:n},G_{0:n})$ if $M_t = RK$.
\end{definition}

To describe how knots act on Feynman--Kac models, we first consider the domain of the requisite operator. Recall the set of all Feynman--Kac models with horizon $n$ as $\fkclass_{n}$. If we let $\knclass$ be the set of all possible knots, we can define the set of all compatible knots and Feynman--Kac models as 
\begin{equation}\label{eq:knotmodelcompat}
        \mathscr{D}_{n} = \{(\kn,\fk) \in \knclass \times \fkclass_{n}: \kn \text{ and } \fk \text{ are compatible}\}
\end{equation}
for $n \in \nn{1}$. We define the knot operator as a right-associative operator acting on elements of this set.
\begin{definition}[Knot operator] \label{def:knotop}
    The knot operator maps compatible knot-model pairs to the space of Feynman--Kac models for horizon $n \in \nn{1}$ and is denoted by ${\ast:\mathscr{D}_{n} \rightarrow \fkclass_{n}}$. We use the infix notation $\kn \ast \fk$ for convenience. For a knot $\kn = (t,R,K)$ and model $\fk = (M_{0:n},G_{0:n})$, the knot-model $\kn \ast \fk = (M_{0:n}^\ast,G_{0:n}^\ast)$ where
    \begin{equation*}
        M_{t}^\ast = R, \quad G_t^\ast = K(G_t), \quad M_{t+1}^\ast = K^{G_t} M_{t+1}.
    \end{equation*}
    The remaining Markov kernels and potential functions are identical to the original model, that is $M_{p}^\ast = M_{p}$ for $p \notin \{t,t+1\}$ and $G_{p}^\ast = G_{p}$ for $p \neq t$.
\end{definition}
The knot operator preserves the terminal predictive and updated measures of the Feynman--Kac model, the predictive and updated path measures (see Proposition~\ref{prop:invar}). Besides preserving key measures, the knot operator preserves the horizon of the Feynman--Kac model it is applied to, which is crucial for our comparisons of the asymptotic variance of particle estimates with and without knots. Figure~\ref{fig:knot} illustrates the transformation of a Feynman--Kac model to a new model by a knot.

\begin{figure}[ht]
    \centering
    \tcbox[left=0.5cm, right=2.25cm]{
        \input{knot-diagram}
    }
    \caption{Effect of a knot $\kn = (t,R_t,K_t)$ at time $t \in \bseq{0}{n-1}$ on model $\fk$ with $M_t=R_t K_t$. Note that $G^\prime_t = K_t(G_t)$.}
    \label{fig:knot}
\end{figure}
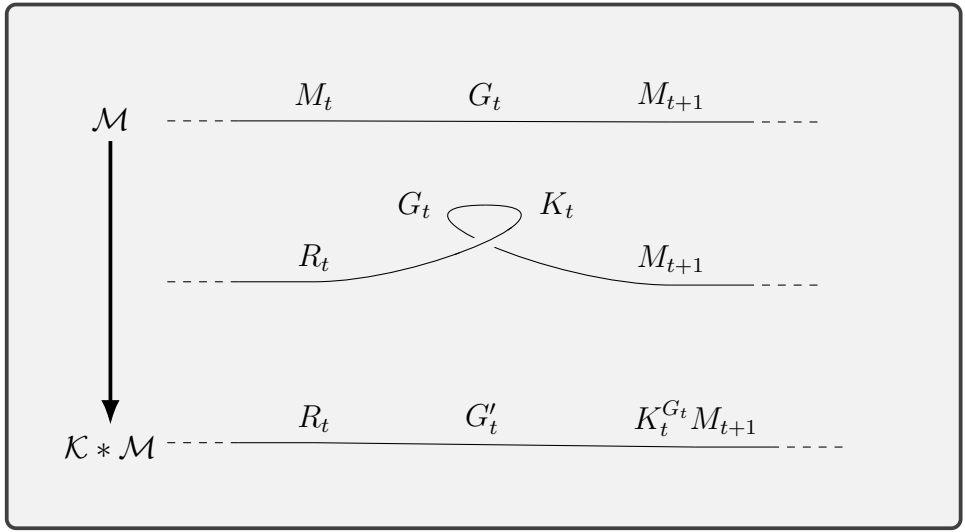
To motivate our consideration of knots, we state a simplified  variance reduction theorem.
\begin{theorem}[Variance reduction from a knot] \label{th:knot}
    Consider models $\fk$ and ${\fk^\ast = \kn \ast \fk}$ for a knot $\kn$. Let~$\fk$ and $\fk^\ast$ have terminal measures $\gamma_n$ and $\gamma_n^\ast$ with asymptotic variance maps $\sigma^2$ and~$\sigma_\ast^2$ respectively. If~$\varphi \in \mathcal{F}(\fk)$ then the terminal measures are equivalent, $\gamma_n(\varphi) = \gamma^\ast_n(\varphi)$, whilst the variances satisfy ${\sigma_\ast^2(\varphi) \leq \sigma^2(\varphi)}$.
\end{theorem}
It is simple to show that Theorem~\ref{th:knot} implies the same asymptotic variance inequality for the marginal updated measures as well as their normalised counterparts. As such, a model with a knot has terminal time particle approximations with better variance properties than the original model. We defer our proof to Theorem~\ref{th:var} which considers the general case with multiple knots.

The simplest possible knot is the trivial knot, in the sense that applying a trivial knot to a Feynman--Kac model does not change the model. The trivial knot is described in Example~\ref{ex:trivialknot}.
\begin{example}[Trivial knot] \label{ex:trivialknot}
    Consider a model $\fk = (M_{0:n},G_{0:n})$ for $n \in \nn1$ and knot~$\kn$ at time $t \in \bseq{0}{n-1}$. If $\kn=(t,M_t,\mathrm{Id})$ then it is trivial in the sense that $\kn \ast \fk = \fk$. 
\end{example}
Trivial knots do not change how the information at time $t$ (the potential $G_t$) is incorporated into the Feynman--Kac model and do not change the asymptotic variance. On the other hand, we can define an adapted knot which fully adapts $M_t$ to the information at time $t$. In fact, any knot can be thought of living on a spectrum between a trivial knot and an adapted knot. We discuss the optimality of adapted knots in Section~\ref{sec:adaptedknots}.
\begin{example}[Adapted knot] \label{ex:adaptedknot}
Consider a model $\fk = (M_{0:n},G_{0:n})$ for $n \in \nn1$ and knot~$\kn$ at time $t \in \bseq{1}{n-1}$. If $\kn=(t,\mathrm{Id},M_t)$ we say it is an adapted knot for $\fk$. The model~$\kn \ast \fk$ has new kernels and potential function, $M_t^\ast = \mathrm{Id}$, $M_{t+1}^\ast = M_t^{G_t} M_{t+1}$, and~${G_{t}^\ast = M_t(G_{t})}$. For $t=0$, an adapted knot has the form $\kn=(0,\delta_0,K_0)$ where the kernel~$K_0$ satisfies $K_0(0,\cdot) = M_0$, whilst $M_0^\ast = \delta_0$, $M_{1}^\ast = M_0^{G_0} M_{1}$, and~${G_{0}^\ast = M_0(G_{0})}$.
\end{example}
At time $t$, an adapted knot results in a kernel of the form $M_{t+1}^\ast = M_t^{G_t} M_{t+1}$ where $M_t^{G_t}$ is now adapted to the information in $G_t$. One might argue that a more natural representation of such an adaptation would use $M_t^{G_t}$ as the $t$th kernel of the new model, not as a component of the $(t+1)$th kernel. However, our definition of knots is precisely what allows for an ordering of the asymptotic variance terms.

\subsection{Knotsets}\label{sec:tyknotset}

A knot is the elementary operator we consider, in the sense that it is the minimal modification of a Feynman--Kac model for which we can prove a variance reduction. However, it is natural to consider a set of knots acting on many time points of a model. Further, knots can be applied sequentially, but it is convenient to consider a set of knots that can be applied simultaneously. As such, we will now define a generalisation of knots and their associated operator, the \textit{knotset} and \textit{knotset operator}.

\begin{definition}[Knotsets and compatibility] \label{def:knotset}
A knotset $\kn=(R_{0:n-1},K_{0:n-1})$ is specified by $n$ knots of the form $\kn_p=(p,R_p,K_p)$ for $p \in \bseq{0}{n-1}$. Such a knotset is compatible with $\fk \in \fkclass_n$ if each $(p,R_p,K_p)$-knot is compatible with $\fk$ for $p\in\bseq{0}{n-1}$.
\end{definition}
 
\begin{definition}[Knotset operator] \label{def:knotsetop}
Let $\kn=(R_{0:n-1},K_{0:n-1})$ be a knotset compatible with~$\fk \in \fkclass_n$. The knotset operation is defined as $\kn \ast \fk = \kn_{0} \ast \kn_{1} \ast \cdots \ast \kn_{n-1} \ast \fk$ where $\kn_p=(p,R_p,K_p)$ for $p \in \bseq{0}{n-1}$.
\end{definition}
The knotset operator is defined to apply $n$ knots, with unique times, in descending order so that the compatibility condition for each knot does not change after each successive knot application. This design also allows us to frame knotsets as a simultaneously application of $n$ knots to a model, which is presented next. 
\begin{proposition}[Knot-model] \label{prop:knotsetop}
    If $\kn = (R_{0:n-1},K_{0:n-1})$ is a knotset compatible with model $\fk = (M_{0:n},G_{0:n})$, then $\kn \ast \fk = (M_{0:n}^\ast,G_{0:n}^\ast)$ satisfies
    \begin{alignat*}{3}
        M_0^\ast &= R_0, &\quad G_0^\ast &= K_0(G_0), & & \\
        M_{p}^\ast &= K_{p-1}^{G_{p-1}}R_p, &\quad G_p^\ast &= K_p(G_p), & &\quad p \in \bseq{1}{n-1} \\
        M_{n}^\ast &= K_{n-1}^{G_{n-1}} M_{n}, &\quad G_n^\ast &= G_n. &
    \end{alignat*}
\end{proposition}
We refer to $\fk^\ast = \kn \ast \fk$ for a knotset $\kn$ as a knot-model and provide the form of~$\fk^\ast$ in Proposition~\ref{prop:knotsetop}. The proof is trivial due to the descending time-ordering of knot applications specified in Definition~\ref{def:knotsetop}. The knotset operator also inherits right-associativity from knots.

We illustrate two examples, trivial and adapted knotsets, that extend Example~\ref{ex:trivialknot}~and~\ref{ex:adaptedknot} respectively. The trivial knotset consists of $n$ trivial knots, as such it does not change the Feynman--Kac model nor alter the asymptotic variance.
\begin{example}[Trivial knotset] \label{ex:trivialknotset}
    Consider a knotset $\kn=(M_{0:n-1},K_{0:n-1})$ where $K_p = \mathrm{Id}$ for $p\in\bseq{0}{n-1}$ applied to model $\fk = (M_{0:n},G_{0:n})$. The resulting model is $\kn \ast \fk = \fk$.
\end{example}
We can also use $n$ adapted knots to form an adapted knotset, which we describe in  Example~\ref{ex:adaptedknot}.
\begin{example}[Adapted knotset] \label{ex:adaptedknotset}
    Consider a knotset $\kn=(R_{0:n-1},K_{0:n-1})$ such that each $(p,R_p,K_p)$-knot is an adapted knot for $\fk = (M_{0:n},G_{0:n})$. The adapted model is $\kn \ast \fk = (M_{0:n}^\ast,G_{0:n}^\ast)$ where $M_0^\ast = \delta_0$, $M_1^\ast(0,\cdot) = M_{0}^{G_{0}}$, $M_{p}^\ast = M_{p-1}^{G_{p-1}}$ for $p\in\bseq{2}{n-1}$, and $M_{n}^\ast = M_{n-1}^{G_{n-1}} M_n$.
    Whilst the potentials satisfy $G_p^\ast = M_p(G_p)$ for $\bseq{0}{n-1}$, and $G_n^\ast = G_n$.
\end{example}
Adapted knotsets are related to fully-adapted auxiliary particle filters \citep{pitt1999filtering,johansen2008note} but differ subtly. We discuss this class of knots and its relation to existing particle filters in Section~\ref{sec:apfperfect}. 

The model in Example~\ref{ex:adaptedknotset} has redundancy since $M_{0}^\ast = \delta_0$ and $M_1^\ast(0,\cdot) = M_{0}^{G_{0}}$ and $G_0^\ast = M_0(G_0)$ is a constant. This is an artifact of the knot operator which preserves the time horizon of the model and is essential for comparing the asymptotic variance terms. Using such the adapted knot-model in practice, one would ignore the initial transitions and begin the particle filter at time $p=1$. If required, the constant potential function at time $p=0$ can be accounted for including it in the potential function at time $p=1$.

Though knotsets can change the Feynman--Kac model they are applied to, some quantities remain unchanged, whilst others can be expressed in terms of the original model. We demonstrate these invariances and equivalences now. In order to distinguish quantities related to the original model, $\fk$, and to the relevant knot-model,~$\fk^\ast$, we  embellish the latter with the same superscript. For example, the terminal measure of $\fk$ will be denoted by $\gamma_n$, whilst we will use $\gamma_n^\ast$ for $\fk^\ast$.
\begin{proposition}[Knot-model predictive measures]\label{prop:predequiv}
    Let $\kn$ be a $(R_{0:n-1},K_{0:n-1})$-knotset and consider knot-model $\fk^\ast = \kn\ast \fk$. For measurable $\varphi$, the knot-model $\fk^\ast$ will have predictive marginal measures such that 
    \begin{enumerate}
        \item $\gamma^\ast_0(\varphi) = R_0(\varphi)$ and $\gamma^\ast_{p}(\varphi) = \hat\gamma_{p-1}R_p(\varphi)$ for $p \in \bseq{1}{n-1}$.
        \item $\gamma^\ast_{p}K_p(\varphi) = \gamma_{p}(\varphi)$ for $p \in \bseq{0}{n-1}$.
    \end{enumerate} 
\end{proposition}
Aside from establishing the connection between a model and its counterpart with knots, Part 1 of Proposition~\ref{prop:predequiv} is later used to compare the asymptotic variance of particle approximations from the use of each Feynman--Kac model in an SMC algorithm. Part 2 indicates that if any marginal measures in the original model are of interest we can approximate these with one additional step, even when using the model with knots. 
\begin{proposition}[Knot-model invariants]\label{prop:invar}
Let $\kn$ be a $(R_{0:n-1},K_{0:n-1})$-knotset and consider knot-model $\fk^\ast = \kn\ast \fk$. For measurable $\varphi$, the knot-model $\fk^\ast$ will have the following invariants:
\begin{enumerate}
    \item Terminal marginal measures, $\gamma^\ast_{n}(\varphi) = \gamma_{n}(\varphi)$ and $\hat\gamma^\ast_{n}(\varphi) = \hat\gamma_{n}(\varphi)$.
    \item Terminal probability measures, $\eta^\ast_{n}(\varphi) = \eta_{n}(\varphi)$ and $\hat\eta^\ast_{n}(\varphi) = \hat\eta_{n}(\varphi)$.
    \item Normalising constants, $\gamma^\ast_{p}(1) = \gamma_{p}(1)$ and $\hat{\gamma}^\ast_p(1) =  \hat\gamma_p(1)$ for all $p \in \bseq{0}{n}$.
\end{enumerate}
\end{proposition}
Proposition~\ref{prop:invar} establishes that the terminal measure is unchanged by knots, hence a model and its counterpart with knots can be used to estimate the same quantities. We will also make use of the invariants when making asymptotic variance comparisons.

In subsequent sections we will use the terms knots and knotset synonymously. We note that a knotset is a strict generalisation of a $(t,R,K)$-knot, which can be seen by taking the underlying knots $\kn_p$ to be trivial for all $p \neq t$ in Proposition~\ref{prop:knotsetop}. As such, results for knotsets apply directly to knots. More practically, we can also use a knotset to describe only $m \in \bseq{0}{n}$ knots by letting~${n-m}$ knots be trivial. This is useful if no suitable knot can be defined for one or more time points.

In Section~\ref{sec:var} we show that terminal particle approximations using $\kn \ast \fk$ have lower asymptotic variance than their counterparts using $\fk$. By the invariance property in Proposition~\ref{prop:invar} this indicates better particle approximations exist for the same quantities of interest when knots can be implemented.

\section{Variance reduction and ordering from knots}\label{sec:var}

Having established the equivalence of all terminal marginal measures in Proposition~\ref{prop:invar}, we now consider the asymptotic variances of particle approximations to these quantities. Our main result is given in Theorem~\ref{th:var} where we state that applying knots to a Feynman--Kac model reduces the variance of particle approximations for all relevant functions. We will denote the predictive and updated asymptotic variance of the knot-model $\fk^\ast$ by~$\sigma^2_{\ast}$ and $\hat{\sigma}^2_\ast$ respectively.

\begin{theorem}[Variance reduction with knots]\label{th:var}
Consider models $\fk \in \fkclass_n$ and ${\fk^\ast = \kn \ast \fk}$ for knotset $\kn = (R_{0:n-1},K_{0:n-1})$.  If~$\varphi \in \mathcal{F}(\fk)$ then ${\sigma^2_\ast(\varphi) \leq \sigma^2(\varphi)}$ and the reduction in the variance is 
\begin{equation*}
    \sigma^2(\varphi) - \sigma^2_\ast(\varphi) = \sum_{p=0}^{n-1}\frac{\gamma_p(1)^2}{\gamma_n(1)^2} \nu_p \left\{\mathrm{Var}_{K_p}\left[ Q_{p,n}(\varphi) \right]\right\},
\end{equation*}
where $\nu_p = \hat\eta_{p-1} R_p$ for $p \in \bseq{1}{n-1}$ and $\nu_0 = R_0$. Moreover, the variance ordering is strict if there exists a time $p \in \bseq{0}{n-1}$ such that $\nu_p \left\{\mathrm{Var}_{K_p}\left[ Q_{p,n}(\varphi) \right]\right\} > 0$.
\end{theorem}
From Theorem~\ref{th:var} we can see that, loosely speaking, the variance is strictly reduced if $Q_{p,n}(\varphi)$ is not constant relative to $K_{p}$. As expected, degenerate $K_p$ do not reduce the variance as we previously stated for the trivial knotset with $K_p = \mathrm{Id}$. 

Note that the variance reduction excludes a contribution from time $n$ due to the absence of a knot at the terminal time. We can also define a terminal time $(n,R,K)$-knot analogously to the knots discussed thus far. However, such a terminal knot will only guarantee a variance reduction of the normalising constant estimate, $\hat\gamma_n(1)$. We discuss terminal knots and how to achieve variance reductions for specific test functions in Section~\ref{sec:tyingterminalknots}.

Theorem~\ref{th:var} is our main result, applying directly to  predictive measures. The variance reduction result is extended to the remaining terminal measures by Corollary~\ref{coro:var2}.
\begin{corollary}[Knot variance reduction with knots]\label{coro:var2}
    Under the conditions of Theorem~\ref{th:var}, the following asymptotic variance inequalities hold.  
    \begin{enumerate}
        \item Predictive terminal probability measure: ${\sigma^2_\ast(\varphi - \eta_{n}^\ast(\varphi)) \leq \sigma^2(\varphi - \eta_{n}(\varphi))}$ if $\varphi \in \mathcal{F}(\fk)$.
        \item Updated terminal measure: ${\hat{\sigma}_\ast^2(\varphi) \leq \hat{\sigma}^2(\varphi)}$ if $\varphi \in \hat{\mathcal{F}}(\fk)$.
        \item Updated terminal probability measure: ${\hat{\sigma}_\ast^2(\varphi-\hat{\eta}^\ast_n(\varphi)) \leq \hat{\sigma}^2(\varphi-\hat{\eta}_n(\varphi))}$ if $\varphi \in \hat{\mathcal{F}}(\fk)$.
    \end{enumerate}
    The inequalities are strict under the same conditions as Theorem~\ref{th:var}.
\end{corollary}
The differences in the asymptotic variances stated in Corollary~\ref{coro:var2} are straightforward to derive using the quantitative result in Theorem~\ref{th:var} so we suppress them here. 

Theorem~\ref{th:var} pertains to variance reduction from the application of one knotset to a Feynman--Kac model, however we can also consider multiple knotsets via iterative application. In doing so, we can establish a partial ordering of Feynman--Kac models induced by knots.
\begin{definition}[A partial ordering of Feynman--Kac with knots]\label{def:partialorder}
    Consider two Feynman--Kac models, $\fk,\fk^\ast\in\fkclass_n$. We say that ${\fk^\ast \preccurlyeq \fk}$ if there exists a sequence of knotsets $\kn_1, \kn_2,\ldots, \kn_m$ such that $\fk^\ast = \kn_m\ast\cdots\ast \kn_1\ast\fk$ for some $m \in \nn{1}$. 
\end{definition}
From the above partial ordering we can state a general variance ordering results for sequential Monte Carlo algorithms. Note that each $\kn_s$ in Definition~\ref{def:partialorder} is required to be compatible with the knot-model resulting from $\kn_{s-1}\ast \cdots \ast \kn_{1}\ast\fk$.
\begin{theorem}[Variance ordering with knots]\label{th:varorder}
      Suppose ${\fk^\ast \preccurlyeq \fk}$ then $\gamma_n^\ast(\varphi) = \gamma_n(\varphi)$, $\hat{\gamma}_n^\ast(\varphi) = \hat{\gamma}_n(\varphi)$, $\eta_n^\ast(\varphi) = \eta_n(\varphi)$, $\hat{\eta}_n^\ast(\varphi) = \hat{\eta}_n(\varphi)$, and the following variance ordering results hold.
      \begin{enumerate}
          \item If $\varphi \in \mathcal{F}(\fk)$ then $\sigma_\ast^2(\varphi) \leq \sigma^2(\varphi)$ and $\sigma_\ast^2(\varphi-\eta_n^\ast(\varphi)) \leq \sigma^2(\varphi-\eta_n(\varphi))$.
          \item If $\varphi \in \hat{\mathcal{F}}(\fk)$ then $\hat{\sigma}_\ast^2(\varphi) \leq \hat{\sigma}^2(\varphi)$ and $\hat{\sigma}_\ast^2(\varphi-\hat{\eta}_n^\ast(\varphi)) \leq \hat{\sigma}^2(\varphi-\hat{\eta}_n(\varphi))$.
      \end{enumerate}
      The inequalities are strict if at least one of the knotsets relating $\fk^\ast$ to $\fk$ satisfy the conditions stated in Theorem~\ref{th:var}.
\end{theorem}
Our partial ordering result allows us to order the asymptotic variance of models related by multiple knots or knotsets. Such a result may be useful for some Feynman--Kac models in practice but will typically be more difficult to implement than a single knot or knotset. The partial ordering is, however, crucial to our exposition and proofs involving knot optimality in Section~\ref{sec:adaptedknots}.

\section{Optimality of adapted knots}\label{sec:adaptedknots}
Adapted knots and knotsets, introduced in Examples~\ref{ex:adaptedknot} and~\ref{ex:adaptedknotset} respectively, possess optimality properties that distinguish them from other knots. Applying an adapted $t$-knot to a Feynman--Kac model results in the largest possible variance reduction of any single $t$-knot. Similarly, an adapted knotset will dominate any other knotset in terms of asymptotic variance reduction. This indicates that adapted knots should be appraised first before considering other types of knots that are compatible with the Feynman--Kac model at hand. The optimality of adapted knots and knotsets is expressed formally in Theorem~\ref{th:adaptknotopt}. 

\begin{theorem}[Single adapted knot optimality]\label{th:adaptknotopt}
If $\kn$ is a knotset (resp. $t$-knot) compatible with~$\fk$, and $\kn^{\diamond}$ is the adapted knotset (resp. adapted $t$-knot) for~$\fk$ then $\kn^{\diamond} \ast \fk \preccurlyeq \kn \ast \fk$.
\end{theorem}
Hence, in conjunction with Theorem~\ref{th:varorder}, the asymptotic variance is lowest with adapted knots and knotsets. Further, from Proposition~\ref{prop:knotsimp} we can state that any sequence of $t$-knots applied to $\fk$ is also dominated by the adapted $t$-knot applied to $\fk$. 

We can also deduce from Theorem~\ref{th:adaptknotopt} that the asymptotic variance can only be reduced beyond that of an adapted $t$-knot by using at least two non-trivial knots. For example, using the adapted $t$-knot followed by some other non-trivial knot at time $s\neq t$ would guarantee a further reduction in the variance.

Next we consider the case of multiple applications of knotsets by comparing models of the form $\kn_m\ast\cdots \ast \kn_1\ast \fk$ to the sequence of adapted knotsets applied to $\fk$.

\begin{theorem}[Multiple adapted knotset optimality]\label{th:madaptknotopt}
Let $m \in \nn1$ and consider two sequences of knotsets, $\kn_t$ and $\kn_t^\star$, over $t \in \bseq{0}{m-1}$. Let $\fk_t = \kn_{t-1} \ast \fk_{t-1}$ and $\fk_{t}^\star = \kn_{t-1}^\star \ast \fk_{t-1}^\star$ for $t \in \bseq{1}{m}$ and initial model $\fk_0 = \fk_0^\star \in \fkclass_n$. If $\kn_t^\star$ is the adapted knotset for $\fk_{t}^\star$ for all $t \in \bseq{0}{m-1}$ then $\fk_m^\star \preccurlyeq \fk_m$.
\end{theorem}
Theorem~\ref{th:madaptknotopt} states that if $\kn_{t}^\star$ are adapted knotsets for $t\in\bseq{1}{m}$, then $\kn_{m}^\star\ast\cdots\ast\kn_1^\star\ast\fk \preccurlyeq \kn_{m}\ast\cdots\ast\kn_1\ast\fk$, and hence a sequence of~$m$ adapted knotsets have a greater variance reduction than any other sequence of $m$ knotsets, $\kn_{m}$ for $t\in\bseq{1}{m}$.

As adapted knotsets reduce the variance by the maximum amount of any knotset in each application, a natural question to ask is; to what extent can the asymptotic variance be reduced by repeated applications of adapted knotsets? Theorem~\ref{th:minimalvar} describes the minimal variance achievable by knots and/or knotsets.
\begin{theorem}[Minimal variance from knots]\label{th:minimalvar}
For every model $\fk \in \fkclass_n$ there exists a sequence of knot(sets) $\kn_n,\kn_{n-1},\ldots,\kn_1$ such that the model
\begin{equation*}
    \fk^\ast = \kn_n \ast \kn_{n-1} \ast \cdots \ast \kn_1 \ast \fk
\end{equation*} 
has asymptotic variance terms satisfying
\begin{equation*}
    v_{n,n}^\ast(\varphi) = v_{n,n}(\varphi),\quad v_{p,n}^\ast(\varphi)=0,\;\text{for}\; p \in \bseq{0}{n-1}
\end{equation*}
    for all $\varphi \in \mathcal{F}(\fk)$ where $v_{p,n}^\ast(\varphi)$ and $v_{p,n}(\varphi)$ are the asymptotic variance terms for $\fk^\ast$ and $\fk$ respectively.
\end{theorem}
Example~\ref{ex:perfadapt} provides a (non-unique) construction of a sequence of knot-models, $\fk_t$ for $t\in\bseq{1}{n}$, where the corresponding $v_{p,n}(\varphi)$ terms are zero for all $p<t$. Hence the model $\fk_n$ proves the existence of a sequence of knots in Theorem~\ref{th:minimalvar}. In fact, $\fk_n$ produces exact samples from the terminal predictive distribution, indicating that repeated applications of knotsets to a Feynman--Kac model can yield a perfect sampler. In an SMC algorithm, the exact samples will be independently and identically distributed only if adaptive resampling \citep{kong1994sequential,liu1995blind,del2012adaptive} is used.
\begin{example}[A sequence of adapted knot-models]\label{ex:perfadapt}
    For $t \in \bseq{0}{n-1}$, consider a sequence of $t$-knots $\kn_{t}^\diamond$ and models $\fk_{t+1} = \kn_{t}^\diamond \ast \fk_{t}$ with initial model $\fk_{0}=(M_{0:n},G_{0:n})$. Let~$\eta_{0:n}$ be the predictive probability measures for $\fk_0$. If $\kn_{t}^\diamond$ is an adapted $t$-knot of $\fk_{t}$  for $t \in \bseq{0}{n-1}$ then $\fk_{t} = (M_{t,0:n},G_{t,0:n})$ such that
    \begin{equation*}
        M_{t,0} = \delta_0, \quad M_{t,t}(0,\cdot) = \eta_{t}, \quad M_{t,p} = \begin{cases}
            \mathrm{Id} &\text{if}~p \in \bseq{1}{t-1}~\text{and}~ t \geq 2,\\
            M_p &\text{if}~p \in \bseq{t+1}{n}~\text{and}~ t \leq n-1,\\
        \end{cases}
    \end{equation*}
for $t \in \bseq{1}{n}$. Whilst the potential functions are
    \begin{equation*}
    G_{t,p} =
    \begin{cases}
        \eta_p(G_p) & \text{if}~p \in \bseq{0}{t-1},\\
        G_{p} & \text{if}~p \in \bseq{t}{n}.
    \end{cases}
    \end{equation*}
\end{example}
The sequence of models $\fk_t$ in Example~\ref{ex:perfadapt} accumulate zero asymptotic variance from times $p\in \bseq{0}{t-1}$ without changing the asymptotic variance in later times $p \in \bseq{t}{n}$. Applying a sequence of adapted knotsets would reduce the overall variance faster and yield the same $\fk_n$ but is more complicated to describe and less informative as an example.

For the final model, $\fk_n$, the remaining variance is only incurred at the terminal time. We can view the SMC algorithm on $\fk_n$ with adaptive resampling as equivalent to an importance sampler using $\eta_n$ as the importance distribution and  weight function $G_n(x_n)\prod_{p=0}^{n-1}\eta_p(G_p) = G_n(x_n)\gamma_{n}(1)$. From the importance sampling view, we know that to reduce the remaining variance its minimal value, we will need to consider both the terminal potential and the test function of interest. Hence we note that terminal knots, introduced next, will need a treatment that reflect this, and is necessarily different from standard knots.

\section{Tying terminal knots}\label{sec:tyingterminalknots}
The knots considered so far have only acted at times $p \in \bseq{0}{n-1}$ and have led to a variance ordering for all terminal particle approximations of relevant test functions. When considering particle approximations of a fixed test function, a terminal knot can be used to (further) reduce the asymptotic variance. Compared to standard knots, terminal knots require special treatment to ensure that the resulting Feynman--Kac model retains the same horizon~$n$ and terminal measure. Naively adapting the knot procedure from Section~\ref{sec:tyknot} would result in a model with an $n+1$ horizon and asymptotic variance that may be difficult to compare. As such, our approach is to explicitly extend the state-space of the terminal time to prepare the Feynman--Kac model for use with terminal knots. We introduce such extended models in Section~\ref{sec:extendedmodels}.

\subsection{Extended models}\label{sec:extendedmodels}
 Any Feynman--Kac model with terminal elements $M_n$ and $G_n$ can be trivially extended by replacing these terminal components with $M_n \otimes \mathrm{Id}$ and $G_n \otimes 1$ respectively. This replacement artificially extends the horizon and preserves the terminal measures, without inducing further resampling events. We generalise this notion in Definition~\ref{def:extfkmodel}, describing a~$\phi$-extension that is useful to characterise variance reduction and equivalence among models with terminal knots for specific test functions.

\begin{definition}[$\phi$-extended Feynman--Kac model]\label{def:extfkmodel}
    Let $\fk = (M_{0:n},G_{0:n})$ and $\phi \in \mathcal{L}(\hat\gamma_n)$ be a $\hat\gamma_n$-a.e.\ positive function. The $\phi$-extended model of $\fk$, $\fk^\phi = (M_{0:n}^\phi,G_{0:n}^\phi)$, has terminal Markov kernel
    $M^\phi_n = M_n \otimes \mathrm{Id}$ where the identity kernel is defined on $(\mathsf{X}_n,\mathcal{X}_n)$, and terminal potential function $G^\phi_n = (G_n\cdot\phi) \otimes \phi^{-1}$.
    The remaining kernels and potentials are unchanged.
\end{definition}
We will refer to $\fk$ as the reference model for the $\phi$-extension and to $\phi$ as the target function. As with the Markov kernels and potentials, marginal measures of $\fk^\phi$ will be distinguished with a $\phi$ superscript. Note that the use of superscript~$\phi$ will be reserved for extended models and should not be confused with twisted Markov kernels or measures. A $\phi$-extended Feynman--Kac model can be thought of as a superficial change to the model, with several equivalences stated next. This construction ensures that the particle approximations are unchanged, and prepares the model for use with terminal knots. It is clear from Definition~\ref{def:extfkmodel} that the non-terminal measures of a $\phi$-extended model are equivalent to that of the reference model. We characterise the equivalences for terminal measures and the asymptotic variance in Proposition~\ref{prop:extequiv}.
\begin{proposition}[$\phi$-extended model equivalences]\label{prop:extequiv}
    Consider the $\phi$-extended model $\fk^\phi$ and reference model $\fk$. 
    Let $\gamma_n$ and $\hat\gamma_n$ be marginal terminal measures of $\fk$ and $\hat\sigma^2$ be the asymptotic variance map. If $\gamma_n^\phi$ and $\hat\gamma_n^\phi$ are the marginal terminal measures of $\fk^\phi$ and $\hat{\sigma}^{2}_\phi$ is the asymptotic variance map then
    \begin{enumerate}
        \item $\gamma_n^\phi(1\otimes \varphi) = \gamma_n(\varphi)$ for all $\varphi \in \mathcal{L}(\gamma_n)$.
        \item $\hat\gamma^\phi_{n}(1\otimes\varphi) = \hat\gamma_{n}(\varphi)$ for all $\varphi \in \mathcal{L}(\hat\gamma_{n})$.
        \item $\hat{\sigma}^{2}_\phi(1 \otimes \varphi) = \hat{\sigma}^{2}(\varphi)$ for all $\varphi \in \hat{\mathcal{F}}(\fk)$.
    \end{enumerate}

\end{proposition}
The $\phi$-extended model creates an additional pseudo-time step in the Feynman--Kac model which can then be manipulated by the terminal knot defined analogously to a standard knot. To work with terminal knots, we will replace reference models with their $\phi$-extended counterpart.

\subsection{Terminal knots}\label{sec:terminalknots}

The definition of a terminal knot is essentially equivalent to that of standard knot in Definition~\ref{def:knot}. However, a knot $(t,R,K)$ will only be terminal with respect to a model $\fk \in \fkclass_n$ when $t=n$. The key difference when applying a terminal knot is the additional compatibility condition on the model described in Definition~\ref{def:terminalkncompat}. In essence, we require an additional time step at $n+1$ for the terminal knot to operate analogously to standard knots, but we do not want an additional resampling step.
\begin{definition}[Terminal knot compatibility] \label{def:terminalkncompat}
Let $\fk^\circ =(M_{0:n}^\circ,G_{0:n}^\circ)$ be a Feynman--Kac model and $\kn=(n,R,K)$ be a terminal knot. The model $\fk^\circ$ and terminal knot $\kn$ are compatible if
\begin{enumerate}[(i)]
    \item The model satisfies $M_{n}^\circ = U \otimes V^{G_n\cdot \phi}$  and $G_n^\circ = V(G_n\cdot \phi) \otimes \phi^{-1}$, for some Markov kernels $U$ and $V$, reference model $\fk$ with terminal potential $G_n$, and target function~$\phi$. 
    \item The knot satisfies $U = RK$.
\end{enumerate}
\end{definition}
Overall, Definition~\ref{def:terminalkncompat} extends the notion of knot-model compatibility to the special case of terminal knots, and hence expands the compatible knot and model set, $\mathscr{D}_{n}$, given in \eqref{eq:knotmodelcompat}. We will refer to compatibility condition~(i) as \textit{model compatibility}.
Recall that the knot operator, $\ast: \mathscr{D}_{n} \rightarrow \fkclass_n$, maps compatible knot-model pairs to the space of Feynman--Kac models for horizon $n \in \nn{0}$. Definition~\ref{def:termknotop} extends this operation to terminal knots.
\begin{definition}[Knot operator, terminal knots] \label{def:termknotop}
     Consider a compatible terminal knot $\kn = (n,R,K)$ and model $\fk = (M_{0:n},G_{0:n})$. If $M_{n} = P_1 \otimes P_2$ and $G_{n} = H \otimes \phi^{-1}$ then the knot operation yields $\kn \ast \fk = (M_{0:n}^\ast,G_{0:n}^\ast)$ where
    \begin{equation*}
        M_{n}^\ast = R \otimes  K^{H} P_2, \quad G_n^\ast = K(H) \otimes \phi^{-1},
    \end{equation*}
    for some Markov kernels $P_1$ and $P_2$, and functions $H$ and $\phi$.
    The remaining Markov kernels and potential functions are identical to the original model, that is $M_{p}^\ast = M_{p}$ and $G_{p}^\ast = G_{p}$ for $p \in \bseq{0}{n-1}$.
\end{definition}

Note that the existence of $P_1$, $P_2$, $H$, and $\phi$ in Definition~\ref{def:termknotop} are guaranteed by compatibility condition~(i), ensuring model compatibility. However, this still leaves the question of which models have the required form. Taking $U = M_n$ and $V=\mathrm{Id}$, demonstrates that a $\phi$-extended model is a compatible model. Further, we can demonstrate the set of compatible models, i.e.\ those satisfying compatibility condition~(i), are closed under application of knots.
\begin{proposition}[Compatible models closed under knots] \label{prop:closedmodelcompat}
    Consider a model $\fk^\circ \in \fkclass_n$ satisfying compatibility condition (i) in Definition~\ref{def:terminalkncompat}. If $\kn$ is a (terminal) knot compatible with $\fk^\circ$ then $\kn \ast \fk^\circ$ will also satisfy the same compatibility condition.
\end{proposition}
Proposition~\ref{prop:closedmodelcompat} demonstrates that compatibility condition~(i) of Definition~\ref{def:terminalkncompat} is preserved by the applications of knots. Hence, models of the form $\kn_m \ast\cdots \ast \kn_1 \ast \fk^\circ$ will satisfy model compatibility, where $\kn_p$ are (terminal) knots or knotsets and $\fk^\circ$ is a $\phi$-extended model. 
Typically though, a terminal knot will be the first to be applied to a $\phi$-extended model, which we demonstrate in Example~\ref{ex:extendterminalknot}.
\begin{example}[Terminal knot for $\phi$-extended model]\label{ex:extendterminalknot}
    If $\fk^\phi$ is the $\phi$-extended model of $\fk= (M_{0:n},G_{0:n})$ and $\kn = (n,R,K)$ is a terminal knot then $\kn \ast \fk^\phi = (M_{0:n}^\ast,G_{0:n}^\ast)$ where $M_{n}^\ast = R \otimes  K^{G_n\cdot\phi}$ and $G_n^\ast = K(G_n \cdot \phi) \otimes \phi^{-1}$.
\end{example}
Example~\ref{ex:extendterminalknot} shows that applying a terminal knot to a $\phi$-extended model incorporates information from the terminal potential function and the target function $\phi$ into the new model. The use of {$\phi$-extensions} and specific target function can be motivated by drawing a comparison to optimal importance distributions (see discussion in Section~\ref{sec:adaptedknots}). These components are carefully constructed to preserve the marginal distributions, at least in some form, and facilitate our variance reduction results in Section~\ref{sec:terminalvar}. 

Proposition~\ref{prop:terminalinvar} states the invariance of terminal measures when a terminal knot is applied.
\begin{proposition}[Terminal knot-model invariants]\label{prop:terminalinvar}
Let $\kn =(n,R,K)$ be a terminal knot and consider knot-model $\fk^\ast = \kn\ast \fk^\circ$. For measurable function $\varphi$, the knot-model $\fk^\ast$ will have the following invariants and equivalences:
\begin{enumerate}
    \item $\gamma^\ast_{p}(\varphi) = \gamma_{p}(\varphi)$, $\eta^\ast_{p}(\varphi) = \eta_{p}(\varphi)$, $\hat\gamma^\ast_{p}(\varphi) = \hat\gamma_{p}(\varphi)$,  and $\hat\eta^\ast_{p}(\varphi) = \hat\eta_{p}(\varphi)$ for all~${p \in \bseq{0}{n-1}}$.
    \item $\hat\gamma^\ast_{n}(1 \otimes \varphi) = \hat\gamma_{n}(1 \otimes \varphi)$ and $\hat\eta^\ast_{n}(1 \otimes \varphi) = \hat\eta_{n}(1 \otimes \varphi)$.
\end{enumerate}
\end{proposition}
Two types of special terminal knots are stated in state in Example~\ref{ex:termtrivialknot}~and~\ref{ex:termadaptedknot} which are the terminal counterparts to trivial and adapted standard knots. Note that for compatibility the reference model will need to be $\phi$-extended before these knots can be applied.
\begin{example}[Trivial terminal knot] \label{ex:termtrivialknot}
    Consider a terminal knot $\kn=(n, P_1,\mathrm{Id})$ and model $\fk \in\fkclass_n$ where the terminal kernel is $M_n = P_1 \otimes P_2$. The model resulting from $\kn \ast \fk = \fk$.
\end{example}
As in the standard case, the trivial terminal knot does not change the Feynman--Kac model. At the other extreme is the adapted terminal knot, for which we discuss optimality in Section~\ref{sec:terminalopt}.
\begin{example}[Adapted terminal knot] \label{ex:termadaptedknot}
    Consider a terminal knot $\kn=(n,\mathrm{Id},P_1)$ and model~$\fk \in \fkclass_n$ where the terminal Markov kernel and potential are $M_n = P_1 \otimes P_2$ and $G_n = H \otimes \phi^{-1}$ respectively. The model $\fk^\ast = \kn \ast \fk$ has new terminal kernel and potential function, 
    \begin{equation*}
        M_{n}^\ast = \mathrm{Id} \otimes  P_1^{H} P_2, \quad G_n^\ast = P_1(H) \otimes \phi^{-1}.
    \end{equation*} Whilst the remaining kernels and potentials are unchanged. Further, if the initial model is a $\phi$-extension, say $\fk^\phi$ where $\fk=(M_{0:n},G_{0:n})$, then 
    \begin{equation*}
        M_{n}^\ast = \mathrm{Id} \otimes  M_n^{G_n\cdot\phi}, \quad G_n^\ast = M_n(G_n \cdot \phi) \otimes \phi^{-1}.
    \end{equation*}
\end{example}

To conclude our treatment of terminal knots, we now define \textit{terminal knotsets}. Unlike standard knotsets, a terminal knotset includes a knot at the terminal time. 
\begin{definition}[Terminal knotsets and compatibility] \label{def:terminalknotset}
A terminal knotset $\kn=(R_{0:n},K_{0:n})$ is specified by $n+1$ knots of the form $\kn_p=(p,R_p,K_p)$ for $p \in \bseq{0}{n}$. Such a knotset is compatible with $\fk \in \fkclass_n$ if each $(p,R_p,K_p)$-knot is compatible with $\fk$ for $p\in\bseq{0}{n}$.
\end{definition}
Note that the compatibility condition for the terminal knot in the set ensures that a $\phi$-extension has occurred on a reference model as part of the construction of $\fk$.
\begin{definition}[Terminal knotset operator] \label{def:terminalknotsetop}
Let $\kn=(R_{0:n},K_{0:n})$ be a terminal knotset compatible with~$\fk \in \fkclass_n$. The terminal knotset operation is defined as $\kn \ast \fk = \kn_{0} \ast \kn_{1} \ast \cdots \ast \kn_{n} \ast \fk$ where~$\kn_p=(p,R_p,K_p)$ for $p \in \bseq{0}{n}$.
\end{definition}
When considering the application of a series of knots to a model, it is not necessary to apply terminal knots first. However, as with standard knotsets, such an order ensures that the individual compatibility conditions correspond to the overall compatibility condition.

A special case of Definition~\ref{def:terminalknotset} is an adapted terminal knotset, which consists of an adapted knotset extended to include an adapted terminal knot. We explore terminal knotsets further in Section~\ref{sec:normalisingconstant} to reduce the variance of normalising constant estimates. 

\subsection{Variance reduction and ordering}\label{sec:terminalvar}

With terminal measure equivalence established in Proposition~\ref{prop:terminalinvar}, we can describe the variance reduction from terminal knots for functions of the form $1 \otimes \varphi$. Recall that the model must be $\phi$-extended before we can apply terminal knots. We start with Theorem~\ref{th:terminaldiffvar}, stating a general result for the difference in asymptotic variance after applying a terminal knot, before carefully specifying the models and test functions for which we can state an asymptotic variance ordering for.

\begin{theorem}[Variance difference with a terminal knot]\label{th:terminaldiffvar}
Consider models $\fk=(M_{0:n},G_{0:n})$ and $\fk^\ast = \kn \ast \fk$ for terminal knot $\kn = (n,R,K)$. If $\bar\varphi = 1\otimes \varphi \in \hat{\mathcal{F}}(\fk)$ then
\begin{equation*}
    \hat\sigma^2(\bar\varphi) - \hat{\sigma}_\ast^2(\bar\varphi) = \frac{\hat\eta_{n-1} R\{\mathrm{Cov}_K(H,{H} \cdot P_2[\{\phi^{-1} \cdot \varphi\}^2])\}}{\eta_n(G_n)^2},
\end{equation*}
where $M_n = P_1 \otimes P_2$ and $G_n = H \otimes \phi^{-1}$.
\end{theorem}
The equivalence of updated terminal measures under $\fk$ and $\fk^\ast$ in Theorem~\ref{th:terminaldiffvar} for a test function $\bar\varphi$ is stated in Proposition~\ref{prop:terminalinvar}.
Clearly, models satisfying $P_2[\{\phi^{-1}\cdot \varphi\}^2] = 1$ almost surely will have a guaranteed variance reduction and there may be certain model classes where more general conclusions can be made. 
We state some sufficient conditions in Corollary~\ref{coro:terminalvar} to ensure a variance reduction.
\begin{corollary}[Variance reduction with a terminal knot]\label{coro:terminalvar}
Consider model $\fk^\circ \in \fkclass_n$ and terminal knot~$\kn = (n,R,K)$ and let $\fk^\ast = \kn \ast \fk^\circ$. For a reference model $\fk= (M_{0:n},G_{0:n})$ and target function $\phi$, if 
\begin{enumerate}
    \item for some $m \in \nn1$ there exists a sequence of knots $\kn_1,\ldots,\kn_m$ such that $\fk^\circ = \kn_m \ast \cdots \ast \kn_1 \ast \fk^\phi$, and
    \item $\varphi \in \hat{\mathcal{F}}(\fk)$ and $\phi = \vert \varphi \vert$, then
\end{enumerate}
\begin{equation*}
    \hat\sigma_\ast^2(1 \otimes \varphi) \leq \hat\sigma_{\circ}^2(1 \otimes \varphi) \leq \hat\sigma^2(\varphi).    
\end{equation*}
Further, if $\fk^\circ = \fk^\phi$ then $\hat\sigma_{\circ}^2(1 \otimes \varphi) = \hat\sigma^2(\varphi)$,
\begin{equation*}
    \hat\sigma^2(\varphi) - \hat{\sigma}_\ast^2(1 \otimes \varphi) = \frac{\hat\eta_{n-1} R\{\mathrm{Var}_{K}(G_n\cdot\vert \varphi \vert)\}}{\eta_n(G_n)^2},
\end{equation*}
and the inequality $\hat\sigma_\ast^2(1 \otimes \varphi) \leq \hat\sigma^2(\varphi)$ is strict if $\hat\eta_{n-1} R\{\mathrm{Var}_{K}(G_n\cdot\vert \varphi \vert)\} > 0$.
\end{corollary}
Corollary~\ref{coro:terminalvar} presents the incremental variance reduction from a terminal knot applied to a model that can be expressed as a $\phi$-extended model with or without knots. It is written to emphasise the case of updated measures, since terminal knots are defined for an updated measure by convention, but includes predictive measures as a special case when $G_n = 1$. Multiple applications of terminal and standard knots are treated by the partial ordering described in Theorem~\ref{th:terminalvarorder}. Importantly, we can only consider $\varphi$ that are almost everywhere non-zero due the the conditions imposed by the $\phi$-extension in Definition~\ref{def:extfkmodel}.

To reduce the variance for a particle approximation to a probability measure, i.e.\ $\hat\eta_n(\varphi)$, Corollary~\ref{coro:terminalvar} implies that one should set $\phi = \vert \varphi  - \hat\eta_n(\varphi)\vert$. However, this would require knowledge of $\hat\eta_n(\varphi)$ in advance. Iterative schemes could be used to approximate such a terminal knot, but the result would be approximate. We leave investigation of such iterative schemes for future work. As present, terminal knots are most amenable to normalising constant estimation which we consider specifically in Section~\ref{sec:normalisingconstant}.

Terminal knots can be used in conjunction with standard knots, but such a combination will only guarantee a reduction in the asymptotic variance for the target function $\phi = \vert \varphi \vert$. Equipped with terminal knots, we can define a partial ordering on Feynman--Kac models specifically for a test function $\varphi$. 
\begin{definition}[A partial ordering of Feynman--Kac with terminal knots]\label{def:terminalpartialorder}
    Consider two Feynman--Kac models, $\fk^\circ,\fk^\ast\in\fkclass_n$ and target function $\phi$. We say that ${\fk^\ast \preccurlyeq_\phi \fk^\circ}$ with respect to a reference model $\fk \in\fkclass_n$ if for some $m,m^\prime \in \nn{1}$ 
    \begin{enumerate}
        \item there exists a sequences of knots $\kn_{1}, \kn_{2}, \ldots, \kn_{m}$  such that $\fk^\ast = \kn_{m}\ast\cdots\ast \kn_{1}\ast\fk^\circ$,
        \item there exists a sequences of knots $\kn_{1}^\prime, \kn_{2}^\prime, \ldots, \kn_{m^\prime}^\prime$  such that 
    $\fk^\circ = \kn_{m}^\prime\ast\cdots\ast \kn_{1}^\prime\ast\fk^\phi$.
    \end{enumerate}
    Each knot in the sequences can be a terminal knots or a standard knot.
\end{definition}
Compared to Definition~\ref{def:partialorder}, this partial ordering now includes terminal knots but at the expense of generality: We are now tied to a single test function $\varphi$ that satisfies $\phi = \vert \varphi \vert$ as the variance ordering states in Theorem~\ref{th:terminalvarorder}. 
\begin{theorem}[Variance ordering with terminal knots]\label{th:terminalvarorder}
      Consider a reference model $\fk \in \fkclass_n$ and let $\varphi \in \hat{\mathcal{F}}(\fk)$. If ${\fk^\ast \preccurlyeq_\phi \fk^\circ}$ with respect to~$\fk$  and $\phi = \vert \varphi \vert$ then 
      \begin{equation*}
        \hat{\gamma}_n^\ast(1 \otimes \varphi) = \hat{\gamma}_n^\circ(1 \otimes \varphi) = \hat\gamma_n(\varphi) \quad \text{and} \quad \hat\sigma_\ast^2(1\otimes\varphi) \leq \hat\sigma_\circ^2(1 \otimes \varphi)  \leq \hat\sigma^2(\varphi).    
      \end{equation*}
\end{theorem}

\subsection{Optimality of adapted terminal knots}\label{sec:terminalopt}
Analogously to their standard counterparts, applying an adapted terminal knot results in the largest variance reduction of any single terminal knot for the test function $\varphi$. We state this result in Theorem~\ref{th:terminaladaptknotopt}.

\begin{theorem}[Adapted terminal knot optimality]\label{th:terminaladaptknotopt}
    Consider a model $\fk^\circ$ satisfying Part~2 of Definition~\ref{def:terminalpartialorder} with reference model $\fk$. Let $\kn$ and $\kn^\diamond$ be terminal knots compatible with $\fk^\circ$. If~$\kn^\diamond$ is the adapted terminal knot for~$\fk^\circ$ then $\kn^\diamond\ast\fk^\circ \preccurlyeq_{\phi} \kn \ast \fk^\circ$ with respect to $\fk$.
\end{theorem}
Beyond this optimality for a single terminal knot, we can also prove that adapted terminal knots allow for the asymptotic variance to be reduced to zero in some cases, in conjunction with standard knots.

\begin{corollary}[Minimal variance from knotsets with terminal knot]\label{cor:terminalminimalvar}
For every model $\fk \in \fkclass_n$ and a.s.\ non-zero $\varphi \in \hat{\mathcal{F}}(\fk)$ there exists a sequence of knots $\kn_{n+1},\kn_{n},\ldots,\kn_1$ such that
\begin{equation*}
    \fk^\star = \kn_{n+1} \ast \kn_{n} \ast \cdots \ast \kn_1 \ast \fk^{\vert \varphi \vert}
\end{equation*} 
has asymptotic variance terms satisfying
\begin{equation*}
    \hat{v}_{n,n}^\star(\bar\varphi) = \hat\eta_n(\vert \varphi \vert)^2 - \hat\eta_n(\varphi)^2,\quad \hat{v}_{p,n}^\star(\varphi)=0,\;\text{for}\; p \in \bseq{0}{n-1}
\end{equation*}
    where $v_{p,n}^\star(\varphi)$ are the asymptotic variance terms for $\fk^\star$ and $\hat\eta_n$ is the terminal updated probability measure for $\fk$.
\end{corollary}
With Corollary~\ref{cor:terminalminimalvar}, we can state that when the target function $\varphi$ is almost surely non-negative or non-positive, then the asymptotic variance of the particle approximation for $\varphi$ is zero. This property is analogous to that of optimal importance functions in importance sampling. We can extended the result to non-negative or non-positive $\varphi$ by replacing the terminal potential by $G_n^0 = G_n \cdot 1_{\overline{S_0}(\varphi)}$. Noting that $\gamma_n(G_n \cdot \varphi) = \gamma_n(G_n^0 \cdot \varphi)$ shows the equivalence, though the terminal probability measures will differ.

To construct particle estimates with zero asymptotic variance under~$\hat\gamma_n$ for more general functions, one could adapt the strategy of ``positivisation'' from importance sampling \citep[see for example,][]{owen2000safe}. We can write the terminal predictive measure of a fixed function $\varphi\in\mathcal{L}(\hat\gamma_n)$ as $\gamma_n(G_n \cdot\varphi) = \gamma_n(G_n^+ \cdot\vert\varphi\vert) - \gamma_n(G_n^- \cdot\vert\varphi\vert)$, where $G_n^+ = G_n \cdot 1_{\varphi>0}$ and $G_n^- = G_n \cdot 1_{\varphi<0}$. From this expression it is natural to consider two SMC algorithms; one with terminal potential $G_n^+$ and the other with $G_n^-$. The underlying Feynman--Kac models and test functions of both algorithms now satisfy the conditions to achieve zero variance.

We now state an example model that achieves the variance in Corollary~\ref{cor:terminalminimalvar} which can be constructed by extending the sequence of models given in Example~\ref{ex:perfadapt}. 

\begin{example}[A sequence of adapted knot-models, continued] \label{ex:terminalperfadapt}
    Consider the sequence of models in Example~\ref{ex:perfadapt} with additional requirement that the initial model $\fk_0 = \fk^\phi$ is a $\phi$-extension such that $\phi = \vert\varphi\vert$. Denote the predictive probability measures of $\fk = (M_{0:n},G_{0:n})$ as $\eta_{0:n}$. Let the next model in the sequence be $\fk_{n+1}=\kn_{n}^\diamond \ast \fk_{n}$. If $\kn_{n}^\diamond$ is the adapted terminal knot for $\fk_{n}$ then the model $\fk_{n+1}=(M_{0:n}^\star,G_{0:n}^\star)$ satisfies
    \begin{equation*}
        M_{0}^\star = \delta_0, \quad M_{n}^\star(0,\cdot) = \delta_0 \otimes \eta_{n}^{G_n\cdot \phi}, \quad M_{p}^\star = \mathrm{Id} ~\text{for}~p\in\bseq{1}{n-1},
    \end{equation*}
with potential functions $G_{p}^\star = \eta_p(G_p)$ for $p \in \bseq{0}{n-1}$ and $G_{n}^\star ={\eta_n(G_{n}\cdot\phi) \otimes \phi^{-1}}$.
\end{example}
Having proven and demonstrated that a target model can be transformed until it has minimal asymptotic variance using knots, we can conclude that
the partial ordering induced by knots includes the optimal model.

\subsection{Estimating normalising constants}\label{sec:normalisingconstant}
One of the most useful cases for terminal knots is when the normalising constant of the Feynman--Kac model, $\hat\gamma_n(1)$, of the model is of primary interest. If $\varphi=1$ is the only test function of interest then it is possible to specify a simplified Feynman--Kac model, which we detail in Example~\ref{ex:termknotsetconstant}, which results from applying a terminal knotset (see Definition~\ref{def:terminalknotsetop}).

\begin{example}[Terminal knotset for normalising constant estimation] \label{ex:termknotsetconstant}
    Consider a $\phi$-extended model $\fk^\phi \in \fkclass_n$ and terminal knotset $\kn = (R_{0:n
    },K_{0:n})$ where $\phi =1$. The knot-model $\fk^\ast = \kn \ast \fk^\phi = (M^\ast_{0:n},G^\ast_{0:n})$ satisfies $M^\ast_{0} = R_0$, $M^\ast_{p} = K_{p-1}^{G_{p-1}} R_p$ for $p\in\bseq{1}{n-1}$, $M_{n}^\ast = K_{n-1}^{G_{n-1}} R_n \otimes K_n^{G_n}$, $G^\ast_{p} = K_p(G_p)$ for $p \in \bseq{0}{n-1}$, and $G_n^\ast = K_n(G_n) \otimes 1$.
     
     The model $\fk^\dagger = (M^\dagger_{0:n},G^\dagger_{0:n})$ with 
\begin{alignat*}{3}
        M_0^\dagger &= R_0, &\quad G_0^\dagger &= K_0(G_0), & & \\
        M_{p}^\dagger &= K_{p-1}^{G_{p-1}} R_p, &\quad G_p^\dagger &= K_p(G_p), & &\quad p \in \bseq{1}{n},
    \end{alignat*}
will satisfy $\hat\gamma_n^\dagger(1) = \hat\gamma_n^\ast(1) =\hat\gamma_n(1)$ and $\hat\sigma^2_\dagger(1) = \hat\sigma^2_\ast(1) \leq \hat\sigma^2(1)$.
\end{example}
Note that $K_n^{G_n}$ does not need to be simulated in the model $\fk^\dagger$. The asymptotic variance for model $\fk^\dagger$ can also be further reduced by applying more knots. A special case of Example~\ref{ex:termknotsetconstant} is using the adapted terminal knotset. In this case, $M_0^\dagger = \delta_0$ and so long as $G_0^\dagger = M_0(G_0)$ is accounted for elsewhere in the algorithm, the first iteration of the SMC algorithm does not need to be run.

\section{Applications and examples}\label{sec:apps}

\subsection{Particle filters with `full' adaptation}\label{sec:apfperfect}

A particle filter with `full' adaptation 
adapts each Markov kernel in the Feynman--Kac model to the current potential information through twisting. Originally proposed as a type of auxiliary particle filter by \citet{pitt1999filtering}, its modern interpretation does away with auxiliary variables, though it is still often referred to as a fully-adapted auxiliary particle filter. It is popular due to its empirical performance and its derivation which is motivated by identifying locally (i.e.\ conditional) optimal proposal distributions for the particle weights at each time step. We refer to this algorithm as a particle filter with `full' adaptation. The Feynman--Kac model for such an algorithm is described in Example~\ref{ex:fullyadaptedpf}.
\begin{example}[Particle filter with `full' adaptation]\label{ex:fullyadaptedpf}
    Let $\fk = (M_{0:n},G_{0:n})$ be a model for a particle filter. The particle filter with `full' adaptation with respect to $\fk$ has model  $\fk^\mathrm{F} = (M_{0:n}^\mathrm{F},G_{0:n}^\mathrm{F})$ satisfying
        \begin{alignat*}{3}
        M_0^\mathrm{F} &= M_0^{G_0}, &\quad G_0^\mathrm{F} &= M_0(G_0) \cdot M_1(G_1), & & \\
        M_{p}^\mathrm{F} &= M_{p}^{G_{p}}, &\quad G_p^\mathrm{F} &= M_{p+1}(G_{p+1}), & &\quad p \in \bseq{1}{n-1} \\
        M_{n}^\mathrm{F} &= M_{n}^{G_{n}}, &\quad G_n^\mathrm{F} &= 1. &
    \end{alignat*}
\end{example}
The adapted knot-model in Example~\ref{ex:adaptedknotset} and the particle filter with `full' adaptation in Example~\ref{ex:fullyadaptedpf} share the same constituent twisted Markov kernels $M_p^{G_p}$ and expected potential functions $M_p(G_p)$, but differ in where these elements are located in time. In particular, $M_{p}^\ast = M_{p-1}^{G_{p-1}} = M_{p-1}^{\mathrm{F}}$ and $G_p^\ast = M_p(G_p) = G_{p-1}^{\mathrm{F}}$ for $p \in \bseq{2}{n-1}$.  A further crucial difference is that the adapted knot-model does not use the terminal twisted kernel $M_n^{G_n}$. Our theory on knots has shown that adapted knot-models order the asymptotic variance for all relevant test functions, whilst \citet{johansen2008note} contained a counterexample for such a result with the particle filter with `full' adaptation. Note that they referred to this model as having `perfect' adaptation.

We now restate the counterexample in Example~\ref{ex:binarypf}, and demonstrate how adapted knot-models guarantee an asymptotic variance reduction whereas the fully-adapted particle filter does not.

\begin{example}[Binary model of \citeauthor{johansen2008note}, \citeyear{johansen2008note}]
    \label{ex:binarypf}
   Let $\mathcal{B}=(M_{0:1},G_{0:1})$ be a Feynman--Kac model with
\begin{equation*}
    M_0(x_0) = 
        \begin{cases}
            \frac{1}{2} & \text{if}~x_0 = 0,\\
            \frac{1}{2} & \text{if}~x_0 = 1,\\
        \end{cases} \qquad
    M_1(x_0, x_1) = 
        \begin{cases}
            1-\delta & \text{if}~x_1 = x_0,\\
            \delta & \text{if}~x_1 = 1-x_0,\\
        \end{cases}
\end{equation*}
respectively, and potential functions
\begin{equation*}
        G_t(x_t) = 
        \begin{cases}
            1-\varepsilon & \text{if}~x_t = y_t,\\
            \varepsilon & \text{if}~x_t = 1- y_t,\\
        \end{cases} \quad t \in \{0,1\}
\end{equation*}
with fixed observations $y_t$ such that $y_0 = 0$ and $y_1 = 1$.
\end{example}
Figure~\ref{fig:adaptrecreate} compares the asymptotic variances of various particle filters (PF) approximating $\hat\eta_1^N(\varphi)$ with $\varphi(x) = x$. We consider the bootstrap PF, PF with `full' adaptation, and adapted knotset PF in Example~\ref{ex:binarypf} with~$\varepsilon = 0.25$ and~$\delta \in (0,1)$, both analytically and empirically. Figure~\ref{fig:adaptrecreate} replicates the second figure in \citet{johansen2008note} with the addition of the adapted knotset PF. The PF with `full' adaptation is specified by Example~\ref{ex:fullyadaptedpf} with $n=1$,
whilst the adapted knotset PF has
$M_0^\ast = \delta_0$, $M_1^\ast(0,\cdot) = M_0^{G_0}M_1$, $G_0^\ast = M_0(G_0)$, and $G_1^\ast = G_1$.  When $\varepsilon = 0.25$, the adapted knotset PF outperforms the other PFs in this regime, whilst the bootstrap PF mostly outperforms the PF with `full' adaptation. The existence of regimes where the bootstrap PF outperforms the PF with `full' adaptation constitutes the counterexample of \citet{johansen2008note}.
\begin{figure}
\centering
\includegraphics[width=0.75\textwidth]{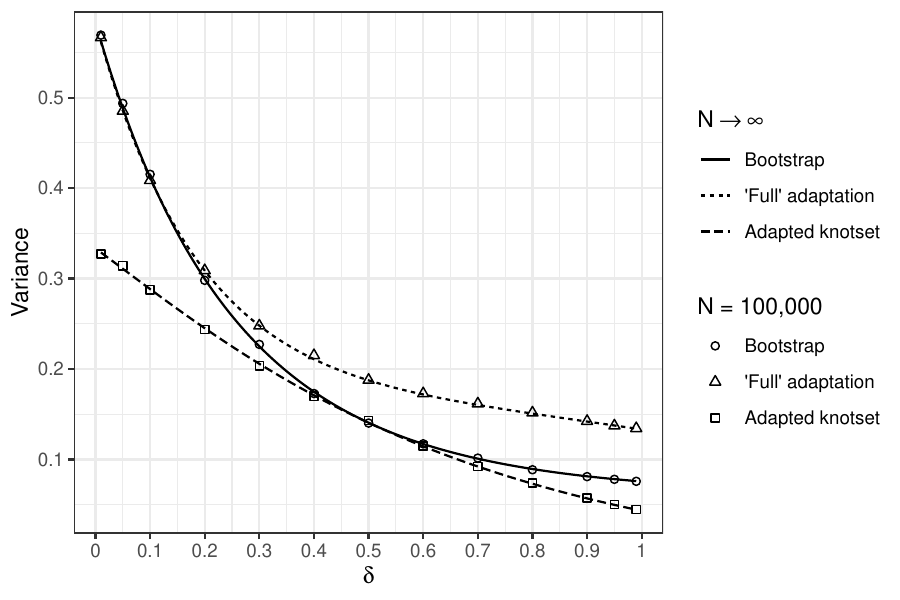}
\caption{Analytical ($N \rightarrow \infty$) and empirical ($N =100{,}000$) asymptotic variance of bootstrap particle filter, filter with `full' adaptation, and adapted knotset particle filter for $\hat\eta_1^N(\varphi)$ in Example~\ref{ex:binarypf} with $\varepsilon=0.25$ and $\varphi(x) = x$. The empirical variance is estimated from 50,000 independent replications.}
\label{fig:adaptrecreate}
\end{figure}

Figure~\ref{fig:adaptvarious} displays the excess analytical asymptotic variance of the PF with `full' adaptation and adapted knotset PF, relative to the bootstrap PF, for $\varepsilon \in \{0.05,0.1,0.2,0.4,0.5\}$ and~$\delta \in (0,1)$. The excess variance of the adapted knotset PF is always less than or equal to zero, demonstrating the dominance over the bootstrap PF, whilst the PF with `full' adaptation can be better or worse than the bootstrap depending on the regime. We also note that the PF with `full' adaptation can outperform the adapted knotset PF for some parameter values and test functions. From our theory, knot-models or equivalent specifications guarantee variance ordering for all relevant test functions when we do not include a knot at the terminal time. Adapting the terminal time, as the PF with `full' adaption does, requires special consideration as discussed in Section~\ref{sec:tyingterminalknots}. Using Example~\ref{ex:binarypf}, we illustrate terminal knots and their relation to the PF with `full' adaptation next.
\begin{figure}
\centering
\includegraphics[width=0.99\textwidth]{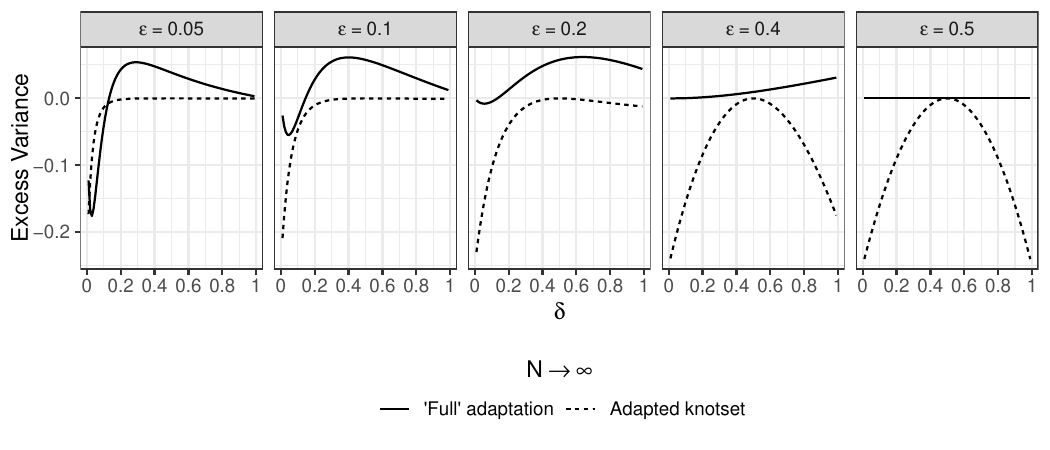}
\caption{Analytical ($N \rightarrow \infty$) asymptotic excess variance of filter with `full' adaptation and adapted knotset particle filter relative to the bootstrap particle filter for $\hat\eta_1^N(\varphi)$ in Example~\ref{ex:binarypf} with $\varepsilon=0.25$ and $\varphi(x) = x$. The excess variance is calculated by subtracting the asymptotic variance of bootstrap particle filter. Negative values indicate an improvement over the bootstrap particle filter. }\label{fig:adaptvarious}
\end{figure}

We now consider Example~\ref{ex:binarypf} for normalising constant estimation, comparing the bootstrap PF, PF with `full' adaption, and PF with an adapted terminal knotset. For the latter particle filter we use the simplified version stated in Example~\ref{ex:termknotsetconstant} with adapted knots, yielding components $M_0^\ast = \delta_0$, $M_1^\ast(0,\cdot) = M_0^{G_0}$, $G_0^\ast = M_0(G_0)$, and $G_1^\ast = M_1(G_1)$.
\begin{figure}
\centering
\includegraphics[width=0.9\textwidth]{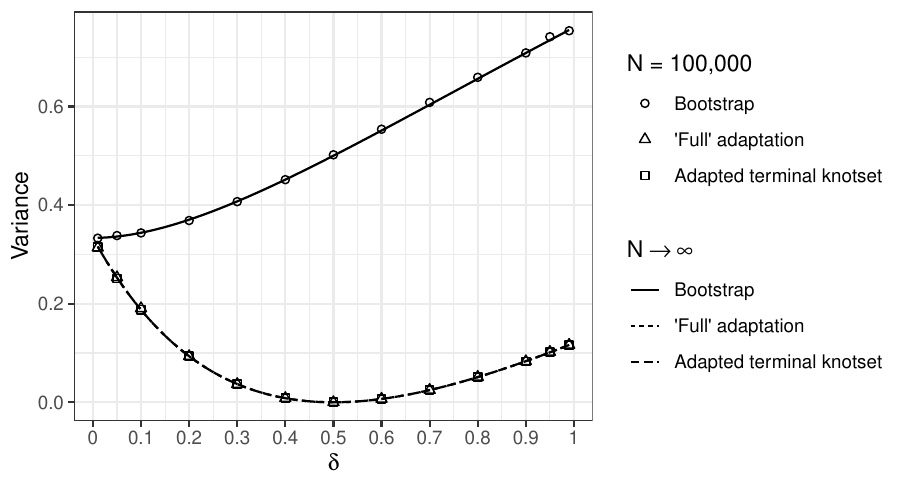}
\caption{Analytical ($N \rightarrow \infty$) and empirical ($N =100{,}000$) asymptotic variance of bootstrap PF, filter with `full' adaptation, and adapted knotset PF for $\hat\gamma_1^N(1)/\hat\gamma_1(1)$ in Example~\ref{ex:binarypf} with $\varepsilon=0.25$. The empirical variance is estimated from 50,000 independent replications. Note that the `full' adaptation and adapted terminal knotset values are theoretically equal.}\label{fig:adaptconstant}
\end{figure}
Figure~\ref{fig:adaptconstant} compares the asymptotic variance of $\hat\gamma_1^N(1)/\hat\gamma_1(1)$ for each particle filter with~$\varepsilon = 0.25$ and~$\delta \in (0,1)$ both analytically and empirically. The figure shows that the variance of the PF with `full' adaptation and adapted terminal knotset PF coincide, as expected as the non-asymptotic estimates are equivalent. As such, the PF with `full' adaptation guarantees variance reduction for normalising constant estimation. 

Overall, our approach to variance reduction with knots explains why the particle filter with `full' adaptation has good empirical performance in many different contexts. Firstly, it is only slightly different to a model (Example~\ref{ex:adaptedknotset}) that we can guarantee an asymptotic variance ordering for all relevant functions. Secondly, it does guarantee a variance reduction for normalising constant estimation because of its equivalence to an adapted terminal knot-model. Thus, we provide a cohesive explanation for the counterexample in \citet{johansen2008note} by clarifying that it is the adaptation to the terminal time that restricts the variance reduction guarantee, not the use of adaptation that is only locally (i.e.\ conditionally) optimal. Further, we have demonstrated how a minor modification to the particle filter with `full' adaptation, Example~\ref{ex:adaptedknotset}, guarantees variance ordering for all relevant test functions which has practical significance. 

\subsection{Marginalisation as a knot}

Model marginalisation in SMC is a well-known variance reduction technique that can be viewed as a special case of knots. Often referred to as Rao--Blackwellisation, the procedure involves analytically marginalising part of the state-space of the Feynman--Kac model, thus reducing the dimensionality of the estimation problem. Historically, Rao--Blackwellisation has been applied to models where it is applicable at all time points, and justified in the case of sequential importance samplers by appealing to a reduction in the variance of the weights \citep{doucet1998sequential,doucet2000sequential}. However, this justification does not relate Rao--Blackwellisation to the variance of particle approximations from a modern SMC algorithm which, in comparison, uses resampling. Framing model marginalisation as a knot operation, we prove this procedure will order the asymptotic variance of SMC algorithms for all relevant test functions when a knot is applied to a non-terminal time point. 

We describe the marginalisation knot in Example~\ref{ex:rbknot} and the application of such a knot in Example~\ref{ex:rbmodel} to demonstrate how certain model assumptions recover existing Rao--Blackwellisation in the literature.
\begin{example}[Marginalisation knot]\label{ex:rbknot} 

Consider a Markov kernel $M_t = U \otimes V$ such that $U:(\mathsf{X}_{t-1},\mathcal{Z}_1) \rightarrow [0,1]$ and $V:(\mathsf{Z}_1 \times \mathsf{X}_{t-1},\mathcal{Z}_2)\rightarrow [0,1]$ for measurable spaces $(\mathsf{Z}_1,\mathcal{Z}_1)$ and $(\mathsf{Z}_2,\mathcal{Z}_2)$.
The knot $\kn = (t,R,K)$ where 
\begin{equation*}
    R(x, \rmd y) = U(x,\rmd y_1)\delta_x(\rmd y_2),\quad  K(y,\rmd z) = \delta_{y_1}(\rmd z_1) V(y,\rmd z_2)
\end{equation*}
is a marginalisation knot and $M_t = RK$.
\end{example}
The kernels $U$ and $V$ partition the state-space $M_t$ is defined on. In particular, $U$ is the marginal distribution of the first component of the partition and $V$ is the conditional distribution of the second component of the partition. The result of applying a marginalisation knot to a model is considered next.
\begin{example}[Model with marginalistion knot]\label{ex:rbmodel}
    Consider $\fk = (M_{0:n},G_{0:n})$ where $M_t = U \otimes V$ with $U,V$, and $\kn = (t,R,K)$ defined as in Example~\ref{ex:rbknot}. If $\kn \ast \fk = (M_{0:n}^\ast,G_{0:n}^\ast)$ then
    \begin{equation*}
        M_{t}^\ast = R, \quad G_t^\ast(y) = V[F(y_1)](y), \quad M_{t+1}^\ast = K^{G_t} M_{t+1},
    \end{equation*}
    where $F$ is a functional such that $F(z_1) = G_t([z_1,\cdot])$ and $K^{G_t}(y,\rmd z) = \delta_{y_1}(\rmd z_1)\otimes V^{F(y_1)}(y,\rmd z_2)$.
\end{example}
Example~\ref{ex:rbmodel} generalises several existing particle filters using marginalised models.
\citet{doucet2000sequential} consider the case where $M_{t+1}(z,\cdot) = P(z_1,\cdot)$, for some kernel $P$, and hence $M_{t+1}$ does not depend on $z_2$. As such, the twisted kernel $V^{F(y_1)}$ is not necessary for particle filter implementation in practice. 
\citet{andrieu2002particle} present the case where $G_t(z) = H(z_1)$, for some potential function $H$ not depending on $z_2$ so that $G_t^\ast = G_t$ and $M_{t+1}^\ast = KM_{t+1}$. The authors assume $M_t$ is Gaussian and apply the Kalman filter to calculate the form of the appropriate marginal and conditional distributions, $U$ and $V$ respectively. Extending this further, \citet{schon2005marginalized} use a Kalman filter to marginalise the linear-Gaussian component of more general state-space models. Special cases include mixture Kalman filters \citep{chen2000mixture} and model-marginalised particle filters for jump Markov
linear systems \citep{doucet2002particle}.
For these examples, and any general (e.g.\ non-Gaussian) state-space model, this paper contributes a complete analysis of asymptotic variance reduction for terminal particle approximations arising from SMC when analytical marginalisation can be performed.

It is also instructive to note that a knot itself can be seen as model marginalisation. If $M_t = RK$ we could extend the Feynman--Kac model from $\fk=(M_{0:n},G_{0:n})$ to $\fk^\prime = (M_{0:n+1}^\prime,G_{0:n+1}^\prime)$ where $M_{t}^\prime = R$, $G_{t}^\prime = 1$, $M_{t+1}^\prime = K$, $G_{t+1}^\prime = G_t$, and $M_{p+1}^\prime = M_{p}$ with $G_{p+1}^\prime = G_{p}$ for $p \in \bseq{t+1}{n}$. Marginalising the state $X_{t+1}^\prime \sim K(x_{t}^\prime,\cdot)$ in $\fk^\prime$, and collecting the remaining terms, results in the model $\kn\ast\fk$ where $\kn = (t,R,K)$. As such, a knot can also be viewed as marginalisation, in particular model extension followed by marginalisation. Our procedures and theory present the most general case of this, as well as nuance around the use of knots at the terminal time. 

\subsection{Non-linear Student distribution state-space models}

In this section we provide a numerical example to illustrate the use of knots in practice, and elucidate the connection between adapted knots and marginalisation knots. We consider a non-linear state-space model with latent variable driven by additive Student noise. The model uses non-linear functions $f_p:
\mathbb{R}^d\rightarrow \mathbb{R}^d$ such that the latent space evolution can be described as 
\begin{align*}
    X_0\sim \mathcal{T}(\mu,\Sigma,\nu),\quad (X_p \mid X_{p-1}) &\sim \mathcal{T}(f_p(X_{p-1}),\Sigma,\nu)~\text{for}~p\in\bseq{1}{n}
\end{align*}
where $\mathcal{T}(\mu,\Sigma,\nu)$ denotes the multivariate Student's t-distribution with mean $\mu\in\mathbb{R}^d$, positive definite scale matrix $\Sigma\in\text{PD}(\mathbb{R}^{d\times d})$, and degrees of freedom $\nu$. We assume the data are observed with Gaussian noise, that is $(Y_p \mid X_p) \sim \mathcal{N}(X_p,\Sigma^\prime)$ with $\Sigma^\prime\in\text{PD}(\mathbb{R}^{d\times d})$ for $p\in\bseq{0}{n}$.

We will use the fact that a multivariate Student distribution can be represented as scale mixture of multivariate Gaussian distributions with transformed $\chi^2_\nu$ distribution, that is a Chi-squared distribution with $\nu$ degrees of freedom. Noting this construction, the Feynman--Kac form of this state-space model has Markov kernels satisfying
\begin{equation}\label{eq:studentfk}
\begin{aligned}
    M_{0} &= (\delta_\mu \otimes S)K\\
    M_{p}(x_{p-1}, \cdot) &= (\delta_{f_p(x_{p-1})} \otimes S)K~\text{for}~p\in\bseq{1}{n}
\end{aligned}
\end{equation}
where $S$ is a $\chi^2_\nu$ distribution and $K$ is conditionally multivariate normal with $K([z,s],\cdot) = \mathcal{N}(z,\frac{\nu}{s}\Sigma)$. Hence, the knot $\kn_p=(p,R_p,K)$ can be applied to the model where $R_0 = \delta_\mu \otimes S$ and $R_p(x_{p-1},\cdot) = \delta_{f(x_{p-1})} \otimes S$ for $p \in \bseq{0}{n-1}$. If we $\phi$-extend the model, we can also apply the analogous terminal knot $\kn_n$.

To test the variance reductions possible for this model, we compare the bootstrap particle filter and the terminal knotset particle filter in Example~\ref{ex:termknotsetconstant} with knots $\kn_0,\ldots, \kn_n$. We consider the problem of normalising constant estimation and assess performance with $R=200$ repetitions of each particle filter. The particle filter implementations used $N=2^{10}$ particles and adaptive resampling with threshold $\kappa = 0.5$. We test the particle filters for five independent datasets simulated by the data generating process that vary by dimension~$d\in\bseq{1}{5}$. The time horizon is fixed at $n=10$, with initial mean $\mu = 0_d$, degrees of freedom $\nu = 4$, and identity variance matrices $\Sigma=\Sigma^\prime=I_d$. We use a univariate non-linear function~$g_p(x) = \frac{x}{2} + \frac{25x}{1 + x^2} + 8 \cos\left(1.2 p \right)$ \citep[see][and references therein]{kitagawa1996monte} to construct a multivariate  non-linear function $f_p(x) = A[g_p(x_1)\cdots g_p(x_d)]^\top$ with matrix $A \in \mathbb{R}^{d\times d}$ having unit diagonal, off-diagonals elements set to half (when $d\geq 2$), and all other elements set to zero (when $d \geq 3$).
\begin{figure}
\centering
\includegraphics[width=0.95\textwidth]{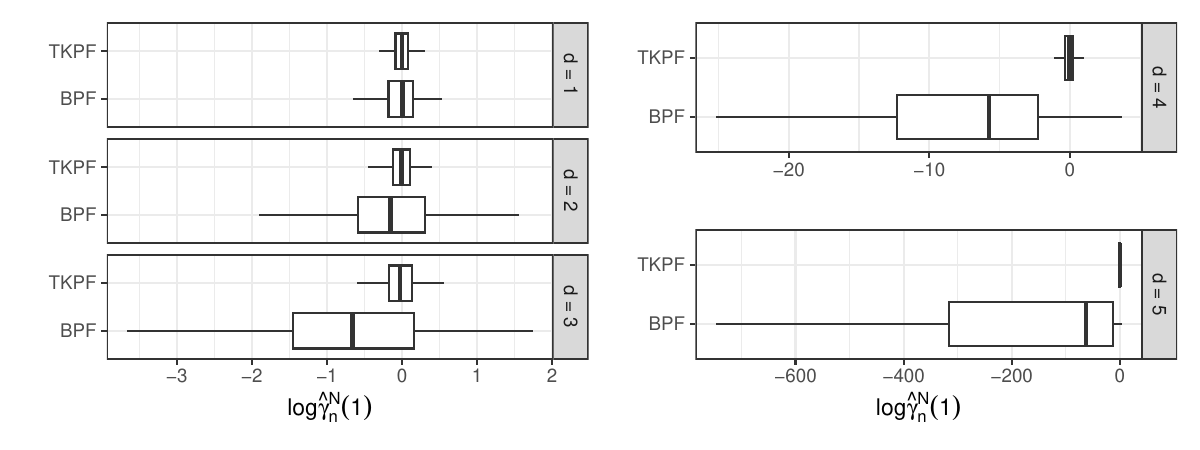}
\caption{Box plots of (shifted) log normalising constant estimates of bootstrap PF and Terminal Knotset PF. The estimates were shifted by a constant (for each dimension $d$) for ease of comparison.}\label{fig:ncstudentbox}
\end{figure}

Figure~\ref{fig:ncstudentbox} displays the estimated normalising constants from each particle filter on the log-scale. For ease of comparison, the log-estimates were shifted by a constant (for each $d$) so that the terminal knotset particle filter estimates have a unit mean (zero on log-scale). The figure demonstrates the variance reduction for normalising constant estimation for each dimension when using knots. We observe that the terminal knotset particle filter remains stable whilst the bootstrap particle filter becomes unstable as $d$ increases.

This example demonstrates a variance reduction for a class of models that appears to be unconsidered in the literature. From the view of adaptation, \citet{pitt1999filtering} showed that `full' adaptation could be implemented with non-linear functions and additive noise. Whilst in this example, we condition on the $\chi^2_\nu$ distributed state and adapt only the Gaussian component in the extended state-space. Whereas from the marginalisation view, the general framework in \citet{schon2005marginalized} does not include the possibility of marginalising a non-linear component, as we do here.

Many further generalisations of this example are possible. For example, a model in the form of \eqref{eq:studentfk} with a $\kn_p$ knot leads to a tractable particle filter any conjugate $K$ and $G_p$ and arbitrary $S$. Such an $S$ could represent an alternative distribution for the scale, or other components of the state-space, and need not be independent of past states.

\section{Discussion}
\label{sec:discussion}

We have shown that knots unify and generalise `full' adaptation and model-marginalisation as variance reduction techniques in sequential Monte Carlo algorithms. Our theory provides a comprehensive assessment of the asymptotic variance ordering implied by knots, the optimality of adapted knots, and demonstrates that repeated applications of adapted knots lead to algorithms with optimal variance. 

In terms of particle filter design, we have re-emphasised the importance of `full' adaptation (or adapted knots) by identifying the pitfall in the counterexample of \citet{johansen2008note}, and explaining how it
can be avoided by not adapting at the terminal time point, or by adapting in a way that takes the test function $\varphi$ into account. Further, given the guaranteed asymptotic variance ordering from knots, we can recommend that every particle filter be assessed to determine if there are one or more tractable knots that can be applied. In such an assessment, the cost of computing $K(G_p)$ and simulating from $K^{G_p}$ should also be considered.

There are several future research directions to explore. Adapted knots have a connection with twisted Feynman--Kac models \citep{guarniero2017iterated,heng2020controlled}. On an extended state-space, a knot can be thought of as decomposing a kernel of a Feynman--Kac model into $R_t$ and $K_t$, adapting $K_t$ to the potential $G_t$, and simplifying the resulting model. Whilst a ``twist'' at time $t$ decomposes the potential function into $\psi_t$ and $\frac{G_t}{\psi}$ and adapts $M_t$ to $\psi_t$. When $K_t = M_t$ and $\psi_t = G_t$ the knot-model and twist-model are equivalent up to a time-shift for $t < n$. Further developing this connection may suggest new methods for twisted Feynman--Kac models in particle filters. Similarly, look-ahead particle filters should also be considered \citep{lin2013lookahead}. For now we note that, except for normalising constant estimation, it may be beneficial for a twisted model to use $\psi_n=1$ as we know that terminal knots do not guarantee an asymptotic variance reduction for all relevant test functions. 

Though we described how terminal knots reduce the asymptotic variance of terminal updated probability measures, we noted a complication arising from the requirement that the target function be $\phi = \vert \varphi  - \hat\eta_n(\varphi)\vert$, where $\hat\eta_n(\varphi)$ is not known ahead of time. It would be possible to estimate $\hat\eta_n(\varphi)$ after the $n-1$ step of the particle filter, and use this to construct the requisite terminal Markov kernel and potential, from the application of a terminal knot to the appropriate $\phi$-extended model, adaptively. Further, it is possible that this procedure could be done for multiple test functions, resulting in an SMC algorithm run until time $n-1$ which then branches for each function of interest. Methodological development of such an algorithm is left for future work.

The application of knots and related variance reductions for SMC samplers is also an area for future research. Consider an SMC sampler with target distribution $\hat\eta_{t-1}$ at time~$t$ with $M_t = RK$. Further, let $R$ be $k$ iterations of some $\hat\eta_{t-1}$-invariant MCMC kernel~$V$, that is $R = V\otimes \cdots \otimes V$, and let $K$ be a random uniform selection over the path generated by $R$. Applying the knot $(t,R,K)$ to the resulting Feynman--Kac model recovers the proposed kernel and potential function used in the recent waste-free SMC algorithm \citep{dau2022waste} at time $t$. Related connections to recent work on Hamiltonian Monte Carlo integrator snippets in an SMC-like algorithm \citep{andrieu2024monte} are also of interest.

\bibliographystyle{chicago}

\bibliography{refs}

\appendix
\section{Proof of main results}\label{app:mainresults}

\subsection*{Proposition~\ref{prop:predequiv}}
\begin{proof} 
Part 1. From Proposition~\ref{prop:knotsetop} $\gamma_0^\ast = R_0$. Further, using the recursion \eqref{eq:marginalrec}, for $p \in \bseq{1}{n-1}$, we have
\begin{equation*}
    \gamma_{p}^\ast(\varphi) = \hat\gamma_{p-1}^\ast M_{p}^\ast(\varphi)
     = \gamma_{p-1}^\ast\{K_{p-1}(G_{p-1}) \cdot K_{p-1}^{G_{p-1}} R_p(\varphi)\}.
\end{equation*}
Then by Proposition~\ref{prop:untwist}, $\gamma_{p-1}^\ast\{K_{p-1}(G_{p-1}) \cdot K_{p-1}^{G_{p-1}} R_p(\varphi)\}
     = \gamma_{p-1}^\ast K_{p-1}\{G_{p-1} \cdot R_p(\varphi)\}$, and hence for $p \in \bseq{1}{n-1}$,
\begin{equation*}
    \gamma_{p}^\ast(\varphi) 
     = \gamma_{p-1}^\ast K_{p-1}\{G_{p-1} \cdot R_p(\varphi)\}.
\end{equation*}
 From this we can see that $\gamma_{1}^\ast(\varphi) = \hat \gamma_{0}R_1(\varphi)$, then $\gamma_{2}^\ast(\varphi) = \hat \gamma_{1}R_2(\varphi)$, and the proof for Part~1 can conclude by induction. 

For Part 2, we use Part 1 to see that $\gamma_0^\ast K_0(\varphi) = R_0 K_0 (\varphi) = M_0(\varphi)=\gamma_0(\varphi)$ and further that
$\gamma^\ast_{p}K_p(\varphi) = \hat\gamma_{p-1}R_pK_p(\varphi)=\hat\gamma_{p-1}M_p(\varphi)=\gamma_{p}(\varphi)$ for $p \in \bseq{1}{n-1}$.
\end{proof}

\subsection*{Proposition~\ref{prop:invar}}

\begin{proof} 
    For Part 1, we expand \eqref{eq:marginalrec} at time $n$ to get 
    \begin{align*}
        \gamma_{n}^\ast(\varphi) &= \gamma_{n-1}^\ast\{G_{n-1}^\ast \cdot M_{n}^\ast(\varphi)\}
        \\
        &= \hat\gamma_{p-2}R_{n-1}\{K_{n-1}(G_{n-1}) \cdot K_{n-1}^{G_{n-1}} M_n(\varphi)\} \\
        &= \hat\gamma_{p-2}R_{n-1}K_{n-1}[G_{n-1} \cdot M_n(\varphi)] \\
        &= \hat\gamma_{p-2}M_{n-1}[G_{n-1} \cdot M_n(\varphi)] \\
        &= \gamma_{n}(\varphi)
    \end{align*}
    from Proposition~\ref{prop:knotsetop}, Proposition~\ref{prop:predequiv}, and Proposition~\ref{prop:untwist}. Since $G_n^\ast = G_n$ the updated terminal measures are also equivalent.
    
    Part 2 follows directly from Part 1 by normalising.
    
    For Part 3, for the predictive measures we have $\gamma^\ast_{n}(1) = \gamma_{n}(1)$ from Part 1. Further, for $p \in \bseq{0}{n-1}$, Proposition~\ref{prop:predequiv} (Part~2) yields $\gamma^\ast_{p}K_p(1) = \gamma_{p}(1)$ then we note that $\gamma^\ast_{p}K_p(1) = \gamma^\ast_{p}(1)$ to gain $\gamma^\ast_{p}(1) = \gamma_{p}(1)$. 

    For the updated measures, for $p \in \bseq{1}{n-1}$, Proposition~\ref{prop:predequiv} (Part~1) yields $\gamma^\ast_{p}(1) = \hat\gamma_{p-1}R_p(1) = \hat\gamma_{p-1}(1)$. Then note that $\gamma^\ast_{p}(1) = \hat\gamma^\ast_{p-1}(1)$ to gain $\hat\gamma^\ast_{p}(1) = \hat\gamma_{p}(1)$ for $p \in \bseq{0}{n-2}$. For $p=n-1$ we have $\hat\gamma^\ast_{n-1}(1) = \gamma^\ast_{n}(1) = \gamma_{n}(1) =  \hat\gamma_{n-1}(1)$ by Part 1, and for $p=n$, $\hat\gamma^\ast_{n}(1) = \hat\gamma_{n}(1)$ again by Part 1.

\end{proof}

\subsection*{Theorem~\ref{th:var}}

\begin{proof} 
    From Proposition~\ref{prop:Qpn} we have $Q_{p,n}^\ast = K_{p} Q_{p,n}$ almost everywhere for $p \in \bseq{0}{n-1}$, and from Proposition~\ref{prop:predequiv} $\gamma_0^\ast = R_0$ and $\gamma_p^\ast = \hat\gamma_{p-1} R_p$ for $p \in \bseq{1}{n-1}$. Therefore, using Jensen's inequality, for $p \in \bseq{1}{n-1}$
    \begin{align*}
        \gamma_p^\ast\{Q_{p,n}^\ast(\varphi)^2\} =\hat\gamma_{p-1} R_p \{K_{p} Q_{p,n}(\varphi)^2\} &\leq \hat\gamma_{p-1} R_p K_{p}\{ Q_{p,n}(\varphi)^2\} = \gamma_{p}\{ Q_{p,n}(\varphi)^2\},~\text{and}\\
        \gamma_0^\ast\{Q_{0,n}^\ast(\varphi)^2\} = R_0\{K_0Q_{0,n}(\varphi)^2\} &\leq R_0K_0\{Q_{0,n}(\varphi)^2\} = \gamma_0\{Q_{0,n}(\varphi)^2\}.
    \end{align*}
     Whilst for $p=n$, and $Q_{n,n}^\ast=Q_{n,n} = \mathrm{Id}$ by definition and~$\gamma_n^\ast=\gamma_n$ from Proposition~\ref{prop:invar} so $\gamma_n^\ast\{Q_{n,n}(\varphi)^2\} = \gamma_n\{Q_{n,n}(\varphi)^2\}$.
    From the above inequalities, and using $\gamma_p^\ast(1) = \gamma_p(1)$ and $\eta_n^\ast = \eta_n$ from Proposition~\ref{prop:invar} we can state that, for $p \in \bseq{0}{n-1}$,
    \begin{equation*}
        v_{p,n}^\ast(\varphi) = \frac{\gamma_p^\ast(1)\gamma_p^\ast(Q_{p,n}^\ast(\varphi)^2)}{\gamma_n^\ast(1)^2} - \eta_n^\ast(\varphi)^2 
        \leq \frac{\gamma_p(1)\gamma_p(Q_{p,n}(\varphi)^2)}{\gamma_n(1)^2} - \eta_n(\varphi)^2 
        = v_{p,n}(\varphi),
    \end{equation*}
    and therefore $\sigma^2_\ast(\varphi) \leq \sigma^2(\varphi)$ by also noting that $v_{n,n}^\ast(\varphi) =v_{n,n}(\varphi)$.

    To quantify the reduction in variance, we can see that 
    \begin{align*}
        \gamma_n\{Q_{n,n}(\varphi)^2\} - \gamma_n^\ast\{Q_{n,n}(\varphi)^2\} &= 0, \\
        \gamma_{p}\{ Q_{p,n}(\varphi)^2\} - \gamma_p^\ast\{Q_{p,n}^\ast(\varphi)^2\} &= \hat\gamma_{p-1} R_p K_{p}\{ Q_{p,n}(\varphi)^2\} -\hat\gamma_{p-1} R_p \{K_{p} Q_{p,n}(\varphi)^2\}\\ 
        &= \hat\gamma_{p-1} R_p \left[ K_{p}\{ Q_{p,n}(\varphi)^2\} -  K_{p} Q_{p,n}(\varphi)^2 \right] \\
        &= \hat\gamma_{p-1} R_p \left[\mathrm{Var}_{K_p}\{Q_{p,n}(\varphi)\} \right],~\text{for}~p\in \bseq{1}{n-1},\\
        \gamma_0\{Q_{0,n}(\varphi)^2\} -\gamma_0^\ast\{Q_{0,n}^\ast(\varphi)^2\} &= R_0K_0\{Q_{0,n}(\varphi)^2\}- R_0\{K_0Q_{0,n}(\varphi)^2\} \\ &= R_0 \left[ K_0\{Q_{0,n}(\varphi)^2\} - K_0Q_{0,n}(\varphi)^2 \right] \\
        &= R_0 \left[\mathrm{Var}_{K_0}\{Q_{0,n}(\varphi)\} \right]
    \end{align*}
    which combined with the measure equivalences with Proposition~\ref{prop:invar} yields the desired result. From this quantification and the original inequality we can conclude that the inequality is indeed strict if the $\nu_p$-averaged variance terms in Theorem~\ref{th:var} are strictly positive.
\end{proof}

\subsection*{Corollary~\ref{coro:var2}}

\begin{proof} 
    For Part 1, if $\varphi \in \mathcal{F}(\fk)$ then $\varphi - \eta_n(\varphi) \in \mathcal{F}(\fk)$ then from Theorem~\ref{th:var} we have $\sigma_\ast^2(\varphi - \eta_{n}^\ast(\varphi)) \leq \sigma^2(\varphi - \eta_{n}(\varphi))$, noting that $\eta_{n}^\ast(\varphi)=\eta_{n}(\varphi) < \infty$.
    
For Part 2, using \eqref{eq:asyvarupdate} the updated asymptotic variance of $\fk^\ast$ can be written as $\hat{\sigma}_\ast^2(\varphi) = \sigma_\ast^2(G_n^\ast\cdot\varphi) / \eta_n^\ast(G_n^\ast)^2$. Then we can state
\begin{equation*}
    \hat{\sigma}_\ast^2(\varphi) = \frac{\sigma_\ast^2(G_n^\ast\cdot\varphi)}{\eta_n^\ast(G_n^\ast)^2} = \frac{\sigma_\ast^2(G_n\cdot\varphi)}{\eta_n(G_n)^2}
\end{equation*}
since we have $G_n^\ast = G_n$ by Proposition~\ref{prop:knotsetop} and $\eta_n^\ast = \eta_n$ by Proposition~\ref{prop:invar}. Lastly, if $\varphi \in \hat{\mathcal{F}}(\fk)$ then $G_n\cdot \varphi \in \mathcal{F}(\fk)$ and so from Theorem~\ref{th:var} $\sigma_\ast^2(G_n\cdot\varphi) \leq \sigma^2(G_n\cdot\varphi)$. Therefore,
\begin{equation*}
    \hat{\sigma}_\ast^2(\varphi) \leq \frac{\sigma^2(G_n\cdot\varphi)}{\eta_n(G_n)^2} = \hat{\sigma}^2(\varphi),
\end{equation*}
and the inequality is strict under the same conditions as Theorem~\ref{th:var}.
Part 3 follows in the same manner as Part 1, but for updated models.
\end{proof}

\subsection*{Theorem~\ref{th:varorder}}

\begin{proof} 
If ${\fk^\ast \preccurlyeq \fk}$ then there exists a sequence of knotsets $\kn_1, \kn_2,\ldots, \kn_m$ such that $\fk^\ast = \kn_m\ast\cdots\ast \kn_1\ast\fk$ for some $m \in \nn{1}$. 
Let $\fk_s = \kn_s \ast \fk_{s-1}$ for $s \in \bseq{1}{m}$ with $\fk_{0} = \fk$ and let the predictive and asymptotic variance maps of $\fk_s$ be $\sigma^2_s$ and $\hat\sigma^2_s$ respectively. We note that $\fk_m = \fk^\ast$. From Theorem~\ref{th:var} and Corollary~\ref{coro:var2} we can state that
$\sigma^2_s(\varphi) \leq \sigma^2_{s-1}(\varphi)$ for $\varphi \in \mathcal{F}(\fk)$ and $\hat\sigma^2_s (\varphi) \leq \hat\sigma^2_{s-1} (\varphi)$ for $\varphi \in \hat{\mathcal{F}}(\fk)$ over $s \in \bseq{1}{m}$. Each inequality will be strict under the same conditions as Theorem~\ref{th:var}. Therefore we can state that
$\sigma_\ast^2(\varphi) = \sigma^2_m(\varphi) \leq \sigma^2_{0}(\varphi) = \sigma^2(\varphi)$ if $\varphi \in \mathcal{F}(\fk)$ and $\hat{\sigma}_\ast^2(\varphi) = \hat{\sigma}^2_m(\varphi) \leq \hat{\sigma}^2_{0}(\varphi) = \hat{\sigma}^2(\varphi)$ if $\varphi \in \hat{\mathcal{F}}(\fk)$. The analogous results for the probability measures follow since $\varphi - \eta_n(\varphi) \in \mathcal{F}(\fk)$ and $\varphi - \hat\eta_n(\varphi) \in \hat{\mathcal{F}}(\fk)$.
\end{proof}

\subsection*{Theorem~\ref{th:adaptknotopt}}

\begin{proof} 
    First consider the case of knots. For $t >0$ and $\kn = (t,R,K)$, let $\mathcal{R} = (t,\mathrm{Id},R)$ then $\mathcal{R} \ast \kn \ast \fk = \kn^{\diamond} \ast \fk$ by Proposition~\ref{prop:knotsimp} and hence $\kn^{\diamond} \ast \fk \preccurlyeq \kn \ast \fk$. For a knot at time $t=0$, $\kn = (0,R,K)$ where $R$ is a probability measure. As such, let $\mathcal{R} = (0,\delta_0,K_0)$ where $K_0(0,\cdot) = R$ to reach the same conclusion.
    
    For knotsets, Proposition~\ref{prop:knotsetcomplete} ensures the existence of a knot $\mathcal{R}$ such that $\mathcal{R} \ast \kn \ast \fk = \kn^\diamond \ast \fk$ for any model $\fk$ and knotset $\kn$. Therefore, $\kn^{\diamond} \ast \fk \preccurlyeq \kn \ast \fk$. 
\end{proof}

\subsection*{Theorem~\ref{th:madaptknotopt}}

\begin{proof} 
    We have that $\fk_m = \kn_{m-1} \ast \fk_{m-1}$ and hence, by Proposition~\ref{prop:knotsetcomplete}, there exists a knotset~$\mathcal{R}_m$ that completes $\kn_{m-1}$, that is~$\mathcal{R}_m \ast \fk_m = \kn_{m-1}^\diamond \ast \fk_{m-1}$ where $\kn_{m-1}^\diamond$ is the adapted knotset for $\fk_{m-1}$. We use Proposition~\ref{prop:repadapteq} to find
    \begin{align*}
        \mathcal{R}_m \ast \fk_m & = \kn_{m-1}^\diamond \ast \kn_{m-2} \ast \cdots \ast \kn_0 \ast \fk_0\\
        &=\mathcal{J}_{m-1} \ast \cdots \ast \mathcal{J}_{1} \ast \kn_0^\star \ast \fk_0 \\
        &= \mathcal{J}_{m-1} \ast \cdots \ast \mathcal{J}_{1} \ast \fk_1^\star,
    \end{align*}
for some knotsets $\mathcal{J}_t$ for $t \in \bseq{1}{m-1}$, where the first equality follows by definition of $\fk_{m-1}$. The process of completing the first knotset (now $\mathcal{J}_{m-1}$) with a $\mathcal{R}_{m-1}$ and changing the adapted knotset in the sequence can be repeated until we have 
\begin{equation*}
    \mathcal{R}_{1} \ast \cdots \ast \mathcal{R}_{m-1} \ast \mathcal{R}_m \ast \fk_m = \mathcal{M}^\ast_m,
\end{equation*}
and hence $\fk_m^\star \preccurlyeq \fk_m$.
\end{proof}

\subsection*{Theorem~\ref{th:minimalvar}}

\begin{proof} 
    The model $\fk_n$ in Example~\ref{ex:perfadapt} satisfies the requirements on the asymptotic variance terms for a sequence of knots. This can be seen by noting all potential functions are constant in this model before the terminal time, hence $v_{p,n}^\ast(\varphi)=0$ for $p \in \bseq{0}{n-1}$. As for the final term we have 
    \begin{align*}
        v_{n,n}^\ast(\varphi) &=  \gamma^\ast_n(G_n^2) - \eta_n^\ast(\varphi)^2 \\
        &= \prod_{p=0}^{n-1}\eta_p(G_p)\eta_n(G_n^2) - \eta_n(\varphi)^2 \\
        &=\gamma_{n}(1) \eta_n(G_n^2) - \eta_n(\varphi)^2 \\
        &= \gamma_n(G_n^2) - \eta_n(\varphi)^2 \\
        &= v_{n,n}(\varphi).
    \end{align*}
As knots are a special case of knotsets, a sequence of knotsets satisfying the variance conditions also exists.
\end{proof}

\subsection*{Proposition~\ref{prop:extequiv}}

\begin{proof} 
    For Part~1, since no elements of the Feynman--Kac model are changed prior to time $n$, we have that $\hat\gamma_{n-1}^\phi = \hat\gamma_{n-1}$ and then
    \begin{equation}\label{eq:extint}
        \gamma_n^\phi(\varphi_1 \otimes \varphi_2) = \hat\gamma_{n-1}M_{n}^\phi(\varphi_1 \otimes \varphi_2) 
        = \hat\gamma_{n-1}M_{n}(\varphi_1 \cdot \varphi_2) 
         = \gamma_{n}(\varphi_1 \cdot \varphi_2).
    \end{equation}
    From this result, we see that $\gamma_n^\phi(1 \otimes \varphi) = \gamma_n(\varphi)$ as required for the predictive measure. As for the updated measure in Part~2 we have 
    \begin{align*}
        \hat\gamma_n^\phi(1 \otimes \varphi) &= \gamma_n^\phi(G_n^\phi \cdot [1 \otimes \varphi]) \\
        &= \gamma_n^\phi([G_n \cdot \phi] \otimes [\varphi \cdot \phi^{-1}]) \\
        &= \gamma_n(G_n \cdot \varphi \cdot \phi \cdot \phi^{-1}) \\
        &= \hat\gamma_n(\varphi \cdot \phi \cdot \phi^{-1}) \\
        &= \hat\gamma_n(\varphi),
    \end{align*}
    using \eqref{eq:extint} for the third equality and the fact that $\phi$ is $\hat\gamma_n$-a.e.\ positive for the final equality.
    For Part~3, starting with the $Q_{p,n}$ terms we first consider $Q_{n}^\phi$ for
\begin{align*}
    Q_{n}^\phi(G_n^\phi \cdot [1\otimes \varphi]) &= G_{n-1}^\phi \cdot  M_n^\phi(G_n^\phi\cdot [1\otimes \varphi]) \\
     &= G_{n-1}^\phi \cdot  (M_n \otimes \mathrm{Id})([G_n \cdot \phi] \otimes [\varphi \cdot \phi^{-1}]) \\
    &= G_{n-1} \cdot M_n(G_n \cdot \varphi \cdot \phi \cdot \phi^{-1}) \\
    &= Q_n(G_n \cdot \varphi),
\end{align*}
almost everywhere w.r.t.\ $\gamma_{n-1}$. Then since $M_p^\phi = M_p$ and $G_p^\phi = G_p$ for $p \in \bseq{0}{n-1}$, we have $Q_p^\phi =Q_p$ for $p \in \bseq{1}{n-1}$, and can state that $Q_{p,n}^\phi(G_n^\phi) = Q_{p,n}(G_n\cdot\varphi)$ almost everywhere under $\gamma_p$ for $p \in \bseq{0}{n-1}$. It is also true that $Q_{n,n}^\phi = \mathrm{Id}$ and $Q_{n,n} = \mathrm{Id}$ by definition, where each identity kernel is defined on their respective measure space. As such, combining with Part~2 we can state that $\hat{v}^\phi_{p,n}(1\otimes\varphi) = \hat{v}_{p,n}(\varphi)$ for $p \in \bseq{0}{n-1}$. 

Whereas for $p=n$, we first note $\{G_n^\phi \cdot [1\otimes\varphi] \}^2 = [G_n \cdot \phi]^2 \otimes [\varphi \cdot \phi^{-1}]^2$, so
\begin{align*}
    M^\phi_n(\{G_n^\phi \cdot [1\otimes\varphi]\}^2) &= (M_n \otimes \mathrm{Id})([G_n \cdot \phi]^2 \otimes [\varphi \cdot \phi^{-1}]^2) \\
    &= M_n([G_n \cdot \varphi]^2),
\end{align*}
almost everywhere under $\hat\gamma_{n-1}$.
Then for the $n$th variance term
\begin{align*}
    \hat{v}_{n,n}^\phi([1 \otimes \varphi]) &= \frac{\gamma_n^\phi(1)\gamma_n^\phi(\{G_n^\phi\cdot [1\otimes\varphi]\}^2)}{\gamma_n^\phi(G_n^\phi)^2} - \eta_n^\phi(G_n^\phi \cdot [1 \otimes \varphi]) \\
    &= \frac{\gamma_n(1)\hat\gamma_{n-1}M_n^\phi(\{G_n^\phi\cdot [1\otimes\varphi]\}^2)}{\gamma_n(G_n)^2} - \eta_n(G_n\cdot [1 \otimes \varphi]) \\
     &= \frac{\gamma_n(1)\hat\gamma_{n-1}M_n(\{G_n\cdot \varphi\}^2)}{\gamma_n(G_n)^2} - \eta_n(G_n\cdot [1 \otimes \varphi])\\
     &= \hat{v}_{n,n}(\varphi),
\end{align*}
using Part~2 for equality of marginal measures.
Hence, $\hat\sigma^2_\phi([1 \otimes \varphi]) = \hat\sigma^2(\varphi)$ since all $p \in \bseq{0}{n}$ terms are equal.
\end{proof}

\subsection*{Proposition~\ref{prop:closedmodelcompat}}

\begin{proof} 
    By compatibility condition (i),  $\fk^\circ = (M_{0:n}^\circ,G_{0:n}^\circ)$ satisfies $M_{n}^\circ = U_1 \otimes V_1^{G_n\cdot \phi}$  and $G_n^\circ = V_1(G_n\cdot \phi) \otimes \phi^{-1}$, for some Markov kernels $U_1$ and $V_1$, reference model $\fk = (M_{0:n},G_{0:n})$, and target function $\phi$. Let $\fk^\ast = \kn \ast \fk^\circ = (M_{0:n}^\ast,G_{0:n}^\ast)$.

     First case. If $\kn$ is a non-terminal knot then $M_{n}^\ast=K^{G_{n-1}}U_1 \otimes V_1^{G_n\cdot \phi}$ and $G_{n}^\ast=V_1(G_n\cdot\phi) \otimes \phi^{-1}$ for some Markov kernel $K$. Letting $U_{2} = K^{G_{n-1}}U_1$ and $V_{2} = V_1$ shows that $\fk^\ast$ satisfies compatibility condition (i).

    Second case. If $\kn$ is a terminal knot then $M_{n}^\ast=R \otimes K^{V_1(G_n\cdot\phi)}V_1^{G_n\cdot \phi}$ and $G_{n}^\ast=KV_1(G_n\cdot\phi) \otimes \phi^{-1}$ for some Markov kernels $R$ and $K$ such that $RK = U_1$. By Proposition~\ref{prop:twisteq}, we can state $M_{n}^\ast=R \otimes (KV_1)^{G_n\cdot \phi}$ and hence letting $U_{2} = R$ and $V_{2} = KV_1$ shows that $\fk^\ast$ satisfies compatibility condition (i).
\end{proof}

\subsection*{Proposition~\ref{prop:terminalinvar}}

\begin{proof} 
For Part 1, since $M_p^\ast = M_p$ and $G_p^\ast = G_p$ for $p \in \bseq{0}{n-1}$ all marginal measures are equal at these times.

By Definition~\ref{def:terminalkncompat} compatibility condition~(i) there exists $P_1$ $P_2$, $H$, and $\phi$ such that $M_{n} = P_1 \otimes P_2$ and $G_{n} = H \otimes \phi^{-1}$. 
Therefore for Part~2 we can consider
\begin{align*}
    \hat\gamma_n^\ast(1\otimes\varphi) &= \hat\gamma_{n-1}^\ast M_n^\ast(G_n^\ast \cdot[1\otimes\varphi]) \\
     &= \hat\gamma_{n-1} (R \otimes K^{H} P_2)(K(H)\otimes [\phi^{-1} \cdot \varphi]) \\
     &= \hat\gamma_{n-1}R\{K(H)\cdot K^{H} P_2(\phi^{-1}\cdot \varphi)\} \\
     &= \hat\gamma_{n-1}RK[H \cdot P_2(\phi^{-1}\cdot \varphi)]
\end{align*}
by Proposition~\ref{prop:untwist}.
Then, by Definition~\ref{def:terminalkncompat} compatibility condition~(ii), we have
\begin{align*}
    \hat\gamma_n^\ast(1\otimes\varphi)&= \hat\gamma_{n-1}P_1[H \cdot P_2(\phi^{-1}\cdot \varphi)] \\
    &= \hat\gamma_{n-1}(P_1 \otimes P_2)([H \otimes \phi^{-1}] \cdot [1 \otimes \varphi]) \\
    &= \hat\gamma_{n-1}M_n(G_n\cdot [1 \otimes \varphi]) \\
    &= \hat\gamma_{n}(1 \otimes \varphi).
\end{align*}
\end{proof}

\subsection*{Theorem~\ref{th:terminaldiffvar}}

\begin{proof} 
For the $Q_{p,n}$ terms we first consider $Q_{n}^\ast$ for
\begin{align*}
    Q_{n}^\ast(G_n^\ast \cdot \bar\varphi) = G_{n-1}^\ast \cdot   M_n^\ast(G^\ast_n\cdot \bar\varphi) 
    = G_{n-1} \cdot  M_n(G_n  \cdot \bar\varphi)
    = Q_n(G_n \cdot \bar\varphi)
\end{align*}
from Proposition~\ref{prop:terminalkerneleq} and noting $G_{n-1}^\ast=G_{n-1}$. Then since $M_p^\ast = M_p$ and $G_p^\ast = G_p$ for $p \in \bseq{0}{n-1}$, we have $Q_p^\ast =Q_p$ for $p \in \bseq{1}{n-1}$, and can state that $Q_{p,n}^\ast(G_n^\ast \cdot \bar\varphi) = Q_{p,n}(G_n\cdot\bar\varphi)$ for $p \in \bseq{0}{n-1}$ and $Q_{n,n}^\ast = Q_{n,n} = \mathrm{Id}$ by definition. As such, combining with Proposition~\ref{prop:terminalinvar} (Part~1),  we can state that $\hat{v}^\ast_{p,n}(\bar\varphi) = \hat{v}_{p,n}( \bar\varphi)$ for $p \in \bseq{0}{n-1}$.  Therefore, we can state that
$\hat\sigma^2(\bar\varphi) - \hat{\sigma}_\ast^2(\bar\varphi) = \hat{v}_{n,n}(\bar\varphi) - \hat{v}_{n,n}^\ast(\bar\varphi)$.

From \eqref{eq:asyvarupdate} we find
\begin{align*}
    \hat{v}_{n,n}(\bar\varphi) - \hat{v}_{n,n}^\ast(\bar\varphi) &= \frac{\eta_n(\{G_n\cdot \bar\varphi\}^2) -\eta_n(G_n\cdot \bar\varphi)^2}{\eta_n(G_n)^2} -\frac{\eta_n^\ast(\{G_n^\ast\cdot \bar\varphi\}^2)-\eta_n^\ast(G_n^\ast\cdot \bar\varphi)^2}{\eta_n^\ast(G_n^\ast)^2}\\
    &= \frac{\hat\eta_{n-1}M_n(\{G_n\cdot \bar\varphi\}^2)}{\eta_n(G_n)^2} - \frac{\hat\eta_{n-1}^\ast M_n^\ast(\{G_n^\ast\cdot \bar\varphi\}^2)}{\eta_n(G_n)^2} \\
    &= \frac{\hat\eta_{n-1}\left[M_n(\{G_n\cdot \bar\varphi\}^2)-M_n^\ast(\{G_n^\ast\cdot \bar\varphi\}^2)\right]}{\eta_n(G_n)^2},
\end{align*}
using $\eta_n^\ast(G_n^\ast\cdot\bar\varphi) = \frac{\hat\gamma_n^\ast(\bar\varphi)}{\gamma_n^\ast(1)} = \frac{\hat\gamma_n(\bar\varphi)}{\gamma_n(1)} = \eta_n(G_n \cdot \bar\varphi)$ and $\hat\eta_{n-1}^\ast = \hat\eta_{n-1}$ from Proposition~\ref{prop:terminalinvar}.

Then we compute the individual terms of the difference, first noting that $M_n = P_1 \otimes P_2$ and $G_n = H \otimes \phi^{-1}$ for some $P_1$, $P_2$, $H$, and $\phi$ by compatibility. 
As such, we have
\begin{align*}
    M^\ast_n(\{G_n^\ast \cdot \bar\varphi\}^2)&= (R \otimes K^{H} P_2)(K({H})^2 \otimes [\phi^{-1} \cdot \varphi]^2) \\
    &= R\{K({H})^2 \cdot K^{H} P_2([\phi^{-1} \cdot \varphi]^2)\} \\
    &= R\{K({H}) \cdot K[ {H} \cdot P_2([\phi^{-1} \cdot \varphi]^2)]\}
\end{align*}
by Proposition~\ref{prop:untwist}, and by simplification
\begin{align*}
    M_n(\{G_n \cdot \bar\varphi\}^2)&= (P_1 \otimes P_2)(H^2 \otimes [\phi^{-1} \cdot \varphi]^2) \\
    &= P_1\{H^2 \cdot P_2([\phi^{-1} \cdot \varphi]^2)\} \\
    &= RK\{H^2 \cdot P_2([\phi^{-1} \cdot \varphi]^2)\}.
\end{align*}
Let $H^\prime = {H} \cdot P_2([\phi^{-1} \cdot \varphi]^2)$ then we can state 
\begin{align*}
    M_n(\{G_n \cdot \bar\varphi\}^2) - M^\ast_n(\{G_n^\ast \cdot \bar\varphi\}^2)
    &= R \left\{K(H \cdot H^\prime) - K({H}) \cdot K( H^\prime)\right\} \\
    &= R\{\mathrm{Cov}_K(H,H^\prime)\},
\end{align*}
completing the proof.
\end{proof}

\subsection*{Corollary~\ref{coro:terminalvar}}

\begin{proof} 
From (repeated application of) Proposition~\ref{prop:closedmodelcompat} there exists Markov kernels $U,V$ such that $M_n^\circ = U \otimes V^{G_n\cdot \phi}$ and $G_n^\circ = V(G_n\cdot\phi) \otimes \phi^{-1}$. With $P_1=U$, $P_2 = V^{G_n\cdot \phi}$, and $H = V(G_n\cdot\phi)$, we note that 
\begin{align*}
    H\cdot P_2[\{\phi^{-1}\cdot \varphi\}^2] &= V(G_n\cdot\phi) \cdot V^{G_n\cdot \phi}[\{\phi^{-1} \cdot \varphi\}^2] \\
    &= V(G_n\cdot \phi) = H,
\end{align*}
since $\phi = \vert \varphi \vert$. Then from Theorem~\ref{th:terminaldiffvar} we can state
\begin{equation}\label{eq:terminalvarres1}
    \hat\sigma^2_\circ(1 \otimes \varphi) - \hat{\sigma}_\ast^2(1 \otimes \varphi) = \frac{\hat\eta_{n-1}^\circ R\{\mathrm{Var}_K(V[G_n \cdot\phi])\}}{\eta_n^\circ(G_n^\circ)}.
\end{equation}
Hence, we have
\begin{equation}\label{eq:terminalvarineq1}
\hat{\sigma}_\ast^2(1 \otimes \varphi) \leq \hat\sigma^2_\circ(1 \otimes \varphi)
\end{equation}
and the inequality will be strict if $\hat\eta_{n-1}^\circ R\{\mathrm{Var}_K(V[G_n \cdot\phi])\} > 0$.

Equation~\eqref{eq:terminalvarineq1} establishes that terminal knots applied to $\phi$-extended models with knots lead to an asymptotic variance ordering for test functions of the form $1 \otimes \varphi$ when $\phi = \vert \varphi \vert$. Since $\fk^\circ$ has this form, every knot in the assumed knot-model sequence reduces the variance of a test function of the form $1 \otimes \varphi$. This follows from Proposition~\ref{th:var} (if a standard knot) or by \eqref{eq:terminalvarineq1} (if a terminal knot). Hence we can state that $\hat{\sigma}_\circ^2(1 \otimes \varphi) \leq \hat{\sigma}_\phi^2(1 \otimes \varphi)$ where $\hat{\sigma}_\phi^2$ is the asymptotic variance map for $\fk^\phi$. Then from Proposition~\ref{prop:extequiv} we have $\hat{\sigma}_\phi^2(1 \otimes \varphi) = \hat{\sigma}^2(\varphi)$ to complete the first part of the proof.

If $\kn$ is the first knot to be applied to $\fk^\phi$ then we have $V = \mathrm{Id}$, $\hat\eta^\circ_{n-1}=\hat\eta^\phi_{n-1}$, and $\hat{\sigma}_\circ^2(1 \otimes \varphi) = \hat{\sigma}^2(\varphi)$ as $\fk^\circ = \fk^\phi$ and $\hat\eta^\phi_{n-1} = \hat\eta_{n-1}$ as $\phi$-extension does not change the non-terminal measures. Using these result in conjunction with \eqref{eq:terminalvarres1} leads to the final result.
\end{proof}

\subsection*{Theorem~\ref{th:terminalvarorder}}

\begin{proof} 
    By definition $\fk^\ast = \kn_m\ast\cdots\ast \kn_1\ast\fk^\circ$ so the variance inequalities follow from iterated applications of Corollary~\ref{coro:terminalvar} (if a terminal knot) or Theorem~\ref{th:var} (if a standard knot).  Further, $\fk^\circ = \kn^\prime_{m^\prime}\ast\cdots\ast \kn^\prime_1\ast\fk^\phi$ by definition so the equalities between measures follows by iterated applications of Proposition~\ref{prop:terminalinvar} (if a terminal knot) and Proposition~\ref{prop:invar} (if a standard knot) to the entire sequence of models. Finally, Proposition~\ref{prop:extequiv} ensures the equivalence of the $\phi$-extended model to the reference model $\fk$.
\end{proof}

\subsection*{Theorem~\ref{th:terminaladaptknotopt}}

\begin{proof} 
    First note that since $\fk^\circ$ satisfies Part~2 of Definition~\ref{def:terminalpartialorder} with respect to $\fk$, by definition, $\kn \ast \fk^\circ$ also satisfies this condition. 
    Then by (repeated application of) Proposition~\ref{prop:closedmodelcompat} there exists Markov kernels $U,V$ such that $M_n^\circ = U \otimes V^{G_n\cdot \phi}$ and $G_n^\circ = V(G_n\cdot\phi) \otimes \phi^{-1}$. Let $\kn=(n,R,K)$ and $\mathcal{R} = (n,\mathrm{Id},R)$, consider the model $\fk^\ast=\mathcal{R}\ast\kn\ast\fk^\circ$, noting that $U = RK$ by compatibility.
    Hence $M_{n}^\ast = \mathrm{Id} \otimes R^{K(H)} K^H V^{G_n\cdot \phi}$ and $G_n^\ast =  RK(H) \otimes \phi^{-1}$ where $H = V(G_n\cdot \phi)$ from the sequential application of the terminal knots.  
    We can simplify the Markov kernel $M_{n}^\ast = \mathrm{Id} \otimes (RKV)^{G_n\cdot \phi} = \mathrm{Id} \otimes(UV)^{G_n\cdot \phi}$ using Proposition~\ref{prop:twisteq} twice. The potential function simplifies to $G_n^\ast =  UV(G_n \cdot \phi) \otimes \phi^{-1}$.

    Now consider the model $\fk^\diamond = \kn^\diamond \ast \fk^\circ$ where $\kn^\diamond$ is the adapted kernel for $\fk^\circ$. The adapted knot is $\kn^\diamond = (n, \mathrm{Id},U)$ and hence $M_{n}^\diamond = \mathrm{Id} \otimes U^{H}V^{G_n\cdot \phi} = \mathrm{Id} \otimes(UV)^{G_n\cdot \phi}$ by Proposition~\ref{prop:twisteq} and $G_n^\ast =  U(H) \otimes \phi^{-1}=UV(G_n \cdot \phi) \otimes \phi^{-1}$. Hence, $M_{n}^\diamond = M_{n}^\ast$ and $G_{n}^\diamond = G_{n}^\ast$, so that we can conclude $\fk^\ast = \fk^\diamond$. 

    Finally, we can state $\fk^\ast=\mathcal{R}\ast\kn\ast\fk^\circ = \kn^\diamond \ast \fk^\circ$ and hence $\kn^\diamond\ast\fk^\circ \preccurlyeq_{\phi} \kn \ast \fk^\circ$ with respect to $\fk$.
\end{proof}

\subsection*{Corollary~\ref{cor:terminalminimalvar}}

\begin{proof} 
Consider $\fk_n$ defined by the sequence in Example~\ref{ex:perfadapt} using initial model $\fk_0 = \fk^{\phi}$ with target function $\phi = \vert \varphi\vert$ and reference model $\fk$. Let $\fk^\star = \kn^\diamond_n \ast \fk_n$ where $\kn^\diamond_n$ is the adapted terminal knot for $\fk_n$. 

First note that from Theorem~\ref{th:minimalvar} we have $\hat{v}_{p,n}^\star(\varphi)=0$ for $p \in \bseq{0}{n-1}$ as the application of the terminal knot $\kn^\diamond_n$ to $\fk_n$ will not change the asymptotic variance terms at earlier times.

From Example~\ref{ex:terminalperfadapt} we can state $\eta_n^\star = \delta_0 \otimes \eta_n^{G_n\cdot\phi}$ and $G_n^\star = \eta_{n}(G_n\cdot\phi) \otimes \phi^{-1}$. Hence, $\eta_n^\star(G_n^\star\cdot \bar\varphi) = (\delta_0 \otimes \eta_n^{G_n\cdot\phi})\{\eta_n(G_n\cdot\phi) \otimes (\phi^{-1}\cdot \varphi)\} = \eta_n(G_n \cdot \varphi)$ and $\eta_n^\star(\{G_n^\star\cdot \bar\varphi\}^2) = (\delta_0 \otimes \eta_n^{G_n\cdot\phi})\{\eta_n(G_n\cdot\phi)^2 \otimes 1\} = \eta_n(G_n\cdot\phi)^2$.

Combining these results with \eqref{eq:asyvarupdate} we find that 
\begin{align*}
    \hat{v}_{n,n}^\ast(\bar\varphi) 
    &= \frac{\eta_n^\ast(\{G_n^\ast\cdot \bar\varphi\}^2) - \eta_n^\ast(G_n^\ast\cdot \bar\varphi)^2}{\eta_n^\ast(G_n^\ast)^2}\\
    &= \frac{\eta_n(G_n\cdot \phi)^2 - \eta_n(G_n\cdot \varphi)^2}{\eta_n(G_n)^2}\\
    &= \hat\eta_n(\phi)^2 - \hat\eta_n(\varphi)^2.
\end{align*}
\end{proof}

\section{Supporting results}\label{app:suppresults}

\begin{proposition}[Kernel untwisting]\label{prop:untwist}
If $K^H:(\mathsf{X},\mathcal{Y})\rightarrow [0,1]$ is a twisted Markov kernel and $\varphi:\mathsf{Y} \rightarrow [-\infty,\infty]$ is measurable then
\begin{equation*}
    K(H) \cdot K^H(\varphi) = K(H\cdot \varphi),
\end{equation*}
$\lambda$-a.e.\ for any measure $\lambda$ on $(\mathsf{X},\mathcal{X})$.
\end{proposition}
\begin{proof} 
Let $\mathsf{X}_+ = \{x\in\mathsf{X}:K(H)(x)>0\}$ and $\mathsf{X}_0 = \{x\in\mathsf{X}:K(H)(x)=0\}$ noting $\mathsf{X} = \mathsf{X}_+ \cup \mathsf{X}_0$ and $\mathsf{X}_+ \cap \mathsf{X}_0 = \emptyset$. 

First case. If $x \in \mathsf{X}_+$ then $K(H)(x) K^H(\varphi)(x) = K(H\cdot \varphi)(x)$. 

Second case. If $x \in \mathsf{X}_0$ we make three observations. Firstly, $K(H)(x) K^H(\varphi)(x) = K(H)(x) K(\varphi)(x)$ by definition. Secondly, $H = 0$ almost surely w.r.t.\ $K(x,\cdot)$ since~$H$ is non-negative. Hence, $K(H\cdot\varphi)(x)=0$ using the standard measure-theoretic convention $0\times\infty=0$ \cite[see for example,][p. 199]{billingsley1995probability}. Lastly, $K(H)(x)K(\varphi)(x)=0$ for $K(\varphi)(x)\in [-\infty,\infty]$ by the same convention, since $K(H)(x)=0$. Therefore, $K(H)(x) K^H(\varphi)(x) = K(H\cdot\varphi)(x) = 0$ for almost every $x \in \mathsf{X}_0$.
\end{proof}

\begin{proposition}[Form of $Q_{p,n}$ with knots]\label{prop:Qpn}
    Suppose $\fk^\ast = \kn \ast \fk$ with $Q_{p,n}^\ast$ terms for~$p \in \bseq{0}{n}$. For any measure $\lambda_p$ on $(\mathsf{Y}_p,\mathcal{Y}_p)$, if $\kn = (R_{0:n-1},K_{0:n-1})$ and $\varphi: \mathsf{X}_n \rightarrow [-\infty, \infty]$ then 
    $Q_{p,n}^\ast(\varphi) = K_{p} Q_{p,n}(\varphi)$ $\lambda_p$-a.e.\ for $p \in \bseq{0}{n-1}$
and $Q_{n,n}^\ast(\varphi) = Q_{n,n}(\varphi)$.
\end{proposition}
\begin{proof}
    Let $R_p$ a probability measure on $(\mathsf{Y}_p,\mathcal{Y}_p)$. Starting with $Q_p^\ast$ terms under $\fk^\ast$, for $p \in \bseq{1}{n-1}$ and $\varphi_p : \mathsf{Y}_p \rightarrow [-\infty, \infty]$ and we have 
\begin{equation}\label{eq:Qpphip}
\begin{aligned}
    Q_p^\ast(\varphi_p) &= K_{p-1}(G_{p-1}) \cdot K_{p-1}^{G_{p-1}} R_p(\varphi_p) \\ &= K_{p-1}\{G_{p-1} \cdot R_p(\varphi_p)\},
\end{aligned}
\end{equation}
$\lambda_{p-1}$-a.e.\ by Proposition~\ref{prop:untwist}. As for $Q_n^\ast$ we have
\begin{align*}
    Q_n^\ast(\varphi) &= K_{n-1}(G_{n-1}) \cdot K_{n-1}^{G_{n-1}} M_n(\varphi) \\ &= K_{n-1}\{G_{n-1} \cdot M_n(\varphi) \} \\ &= K_{n-1}Q_n(\varphi),
\end{align*}
$\lambda_{n-1}$-a.e.\ by Proposition~\ref{prop:untwist}. Similarly, it follows that
\begin{equation*}
    Q_{n-1,n}^\ast(\varphi) = Q_{n}^\ast(\varphi) 
    = K_{n-1}Q_n(\varphi) 
    = K_{n-1}Q_{n-1,n}(\varphi).
\end{equation*}
Now, assume $Q_{p+1,n}^\ast(\varphi) = K_{p+1} Q_{p+1,n}(\varphi)$ for a given $p \in \bseq{0}{n-2}$.
Then by \eqref{eq:Qpphip}
\begin{align*}
     Q_{p,n}^\ast(\varphi) &= Q_{p+1}^\ast Q_{p+1,n}^\ast(\varphi) \\
    &= \int K_{p}(\cdot, \rmd x_{p}) G_{p}(x_{p}) R_{p+1}(x_{p}, \rmd y_{p+1}) K_{p+1}(y_{p+1}, \rmd x_{p+1}) Q_{p+1,n}(\varphi)(x_{p+1}) \\
    &= \int K_{p}(\cdot, \rmd x_{p}) G_{p}(x_{p}) M_{p+1}(x_{p}, \rmd x_{p+1}) Q_{p+1,n}(\varphi)(x_{p+1}) \\
    &= K_p Q_{p+1} Q_{p+1,n}(\varphi) \\
    &= K_p Q_{p,n}(\varphi),
\end{align*}
$\lambda_{p}$-a.e.
Therefore by induction we have
$Q_{p,n}^\ast(\varphi) = K_{p} Q_{p,n}(\varphi)$ $\lambda_{p}$-a.e. for $p \in \bseq{0}{n-1}$,
and $Q_{n,n}^\ast(\varphi) = Q_{n,n}(\varphi)$ since $Q_{n,n}^\ast = Q_{n,n} = \mathrm{Id}$ by definition.
\end{proof}

\begin{proposition}[Twisting kernels equivalence]\label{prop:twisteq}
    Let $R:(\mathsf{X},\mathcal{Y})\rightarrow [0,1]$ and $K:(\mathsf{Y},\mathcal{Z})\rightarrow [0,1]$ be two Markov kernels for measurable spaces  $(\mathsf{Y},\mathcal{Y})$ and $(\mathsf{Z},\mathcal{Z})$, and let $H$ be a non-negative real-valued function. If $R^{K(H)}$ and $K^H$ are twisted Markov kernels then $R^{K(H)} K^H = (RK)^H$ $\lambda$-a.e.\ for any measure $\lambda$ on $(\mathsf{X},\mathcal{X})$.
\end{proposition}

\begin{proof}
    Note that since $K(H)$ is integrable w.r.t.\ $R(x,\cdot)$ by definition and $R\{K(H)\} = RK(H)$ then we have $H$ integrable w.r.t.\ $RK(x,\cdot)$ for $x \in \mathsf{X}$. 
    
    By definition of the twisted kernel we have
    \begin{equation*}
        R^{K(H)} K^H(x, \rmd z) =
        \begin{cases}
            \int_{\mathsf{Y}} \frac{R(x,\rmd y)K(H)(y)}{RK(H)(x)}K^H(y,\rmd z) &\text{if}~ RK(H)(x) > 0, \\
             RK^H(x,\rmd z), &\text{if}~ RK(H)(x) = 0.
        \end{cases}
    \end{equation*}
    In the first case,
    \begin{align*}
        \int_{\mathsf{Y}} \frac{R(x,\rmd y)K(H)(y)}{RK(H)(x)}K^H(y,\rmd z) &= \int_{\mathsf{Y}} \frac{R(x,\rmd y)K(y,\rmd z)H(z)}{RK(H)(x)} \\
        &= \frac{RK(x,\rmd z)H(z)}{RK(H)(x)}
    \end{align*}
    by Proposition~\ref{prop:untwist}. As for the second case, $RK(H)(x) = 0$ implies $K(H)=0$ almost surely w.r.t.\ $R$, hence $RK^H = RK$, by definition of $K^H$. Thus, $R^{K(H)}K^H = (RK)^H$ $\lambda$-a.e.   
\end{proof}

\begin{proposition}[Simplification of two $t$-knots]\label{prop:knotsimp}
    Let $\kn_1 = (t,R_1,K_1)$ and $\kn_2= (t,R_2,K_2)$ be knots for $t \in \bseq{0}{n-1}$. If $\kn_1$ is compatible with $\fk \in \fkclass_n$ and $\kn_2$ is compatible with $\kn_1 \ast \fk$ then $\kn_2\ast \kn_1 \ast \fk = \kn_3 \ast \fk$ where $\kn_3 =(t,R_2,K_2K_1)$.
\end{proposition}
\begin{proof}
    Let $\fk = (M_{0:n},G_{0:n})$. By repeated use of Definition~\ref{def:knotop}, the model $\kn_2\ast \kn_1 \ast \fk = (M_{0:n}^{(2)},G_{0:n}^{(2)})$ has
    \begin{align*}
        M_t^{(2)}=R_2,\quad M_{t+1}^{(2)} = K_2^{K_1(G_t)} K_1^{G_t} M_{t+1}, \quad G_t^{(2)}=K_2K_1(G_t),
    \end{align*}
    and the remaining kernels and potentials are the same as those in $\fk$. By Proposition~\ref{prop:twisteq} we can state $M_{t+1}^{(2)} = (K_2 K_1)^{G_t}M_{t+1}$.
     Therefore the kernels and potential of $\kn_2\ast \kn_1 \ast \fk$ are equal to that of $\kn \ast \fk$ where $\kn =(t,R_2,K_2 K_1)$.
\end{proof}

\begin{proposition}[Completion of a knotset]\label{prop:knotsetcomplete}
    Suppose $\kn$ is a knotset compatible with $\fk$ and $\kn^\diamond$ is the adapted knotset for $\fk$. Then there exists a knotset $\mathcal{R}$ such that $\mathcal{R} \ast \kn \ast \fk = \kn^\diamond \ast \fk$.
\end{proposition}
\begin{proof}
    Consider $\kn = (R_{0:n-1},K_{0:n-1})$ and $\fk = (M_{0:n},G_{0:n})$ and recall $\fk^\ast = \kn \ast \fk$ has Markov kernels $M_0^\ast = R_0$, $M_{p}^\ast = K_{p-1}^{G_{p-1}} R_p$ for $p \in \bseq{1}{n-1}$, and $M_{n}^\ast = K_{n-1}^{G_{n-1}}M_n$, whilst the potentials are $G_p^\ast = K_p(G_p)$ for $p\in \bseq{0}{n-1}$ and $G_n^\ast = G_n$. Let $\mathcal{R} = (U_{0:n-1},V_{0:n-1})$ where $U_p = K_{p-1}^{G_{p-1}}$ and $V_p = R_p$ for $p \in\bseq{1}{n-1}$, whilst $U_0 = \delta_0$ and some $V_0$ satisfying $V_0(0,\cdot) = R_0$. The model $\fk^\prime = \mathcal{R} \ast \fk^\ast = (M_{0:n}^\prime,G_{0:n}^\prime)$ satisfies
    \begin{align*}
        M_0^\prime &= \delta_0,\quad M_{1}^\prime(0,\cdot) = V_{0}^{K_{0}(G_{0})}K_{0}^{G_{0}}, \quad M_{n}^\prime = R_{n-1}^{K_{n-1}(G_{n-1})}K_{n-1}^{G_{n-1}}M_n, \\
        M_{p}^\prime &= R_{p-1}^{K_{p-1}(G_{p-1})}K_{p-1}^{G_{p-1}}~\text{for}~ p \in \bseq{2}{n-1}.
    \end{align*}
    By Proposition~\ref{prop:twisteq} we can state $V_{0}^{K_{0}(G_{0})}K_{0}^{G_{0}} = (V_0K_0)^{G_0}$ and $R_{p}^{K_{p}(G_{p})}K_{p}^{G_{p}} = (R_p K_p)^{G_p}$ for $p \in \bseq{1}{n-1}$. Hence, by compatibility
    \begin{align*}
        M_0^\prime &= \delta_0,\quad M_{1}^\prime(0,\cdot) = M_{0}^{G_{0}}, \quad M_{n}^\prime = M_{n-1}^{G_{n-1}}M_n, \\
        M_{p}^\prime &= M_{p-1}^{G_{p-1}}~\text{for}~ p \in \bseq{2}{n-1}.
    \end{align*}
Whilst the potential functions satisfy  $G_p^\prime = R_p K_p(G_p) = M_p(G_p)$ for $p \in \bseq{0}{n-1}$ and $G_n^\prime = G_n$. Hence, $\fk^\prime = \kn^\diamond \ast \fk$ which completes the proof.
\end{proof}

\begin{proposition}[Adapted knotset equivalence]\label{prop:adapteq}
    Let $\kn$ be a knotset and suppose $\kn^{\diamond}$ is the adapted knotset for $\fk$ and $\kn^{\diamond}_\ast$ is the adapted knotset for $\fk^\ast = \kn \ast \fk$. Then there exists a knotset $\kn^\prime$ such that $\kn^\diamond_\ast \ast \kn \ast \fk = \kn^\prime \ast \kn^\diamond \ast \fk$.
\end{proposition}
\begin{proof}
    Consider $\kn = (R_{0:n-1},K_{0:n-1})$, $\fk = (M_{0:n},G_{0:n})$, and let $\fk_1 = \kn^\diamond_\ast \ast \fk^\ast = \kn^\diamond_\ast \ast \kn \ast \fk $. The model $\fk^\ast = (M_{0:n}^\ast,G_{0:n}^\ast)$ follows Proposition~\ref{prop:knotsetop} and since $\kn^\diamond_\ast$ is the adapted knotset for $\fk^\ast$, we can state $\fk_1 = (M_{1,0:n},G_{1,0:n})$ where the kernels are $M_{1,0} = \delta_0$, $ M_{1,p} = (M_{p-1}^\ast)^{G_{p-1}^\ast}$ for $p \in \bseq{1}{n-1}$, and $ M_{1,n} = (M_{n-1}^\ast)^{G_{n-1}^\ast} M_{n}$, whilst the potentials are $G_{1,p} = M_p^\ast(G_p^\ast)$ for $p \in \bseq{0}{n-1}$ and $G_{1,n} =G_n$. Noting that $R_p K_p = M_p$, $(M_0^\ast)^{G_0^\ast} = R_0^{K_0(G_0)}$,  and $(M_p^\ast)^{G_p^\ast} = (K_{p-1}^{G_{p-1}}R_p)^{K_p(G_p)} = (K_{p-1}^{G_{p-1}})^{M_p(G_p)}R_p^{K_p(G_p)}$ for $p \in \bseq{1}{n-1}$ by Proposition~\ref{prop:twisteq} we can state
\begin{alignat*}{3}
        M_{1,0} &= \delta_0, &\quad G_{1,0} &= M_0(G_0), & & \\
        M_{1,1}(0,\cdot) &= R_{0}^{K_{0}(G_{0})} &\quad G_{1,1} &= K_{0}^{G_{0}}M_1(G_1), & & \\
        M_{1,p} &= (K_{p-2}^{G_{p-2}})^{M_{p-1}(G_{p-1})}R_{p-1}^{K_{p-1}(G_{p-1})} &\quad G_{1,p} &= K_{p-1}^{G_{p-1}}M_p(G_p), & & \quad p \in \bseq{2}{n-1} \\
        M_{1,n} &= (K_{n-2}^{G_{n-2}})^{M_{n-1}(G_{n-1})}R_{n-1}^{K_{n-1}(G_{n-1})} M_{n}, &\quad G_{1,n} &= G_n. & 
    \end{alignat*}
Next, consider $\fk^\diamond = \kn^\diamond \ast \fk = (M^\diamond_{0:n},G^\diamond_{0:n})$, as in Example~\ref{ex:adaptedknotset}, and note that it can be rewritten as 
\begin{alignat*}{3}
        M_{0}^\diamond &= \delta_0, &\quad G_{0}^\diamond &= M_0(G_0), & & \\
         M_{1}^\diamond(0,\cdot) &= R_{0}^{K_{0}(G_{0})}K_{0}^{G_{0}}, &\quad G_{p}^\diamond &= M_p(G_p), & & \\
        M_{p}^\diamond &= R_{p-1}^{K_{p-1}(G_{p-1})}K_{p-1}^{G_{p-1}}, &\quad G_{p}^\diamond &= M_p(G_p), & &\quad p \in \bseq{2}{n-1} \\
        M_{n}^\diamond &= R_{n-1}^{K_{n-1}(G_{n-1})}K_{n-1}^{G_{n-1}} M_{n}, &\quad G_{n}^\diamond &= G_n &
    \end{alignat*}
as $M_p^{G_p} = R_p^{K_p(G_p)} K_P^{G_p}$ by Proposition~\ref{prop:twisteq} and since $R_p K_p=M_p$. Finally, let $\kn^\prime = (R_{0:n-1}^\prime,K_{0:n-1}^\prime)$ and $\fk_2 = \kn^\prime \ast \fk^
\diamond$ where $R_0^\prime = \delta_0$ and $K_0^\prime = \mathrm{Id}$, $R_p^\prime = R_{p-1}^{K_{p-1}(G_{p-1})}$ and $K_p^\prime = K_{p-1}^{G_{p-1}}$ for $p \in \bseq{1}{n-1}$. Therefore from Proposition~\ref{prop:knotsetop}, $\fk_2 = (M_{2,0:n},G_{2,0:n})$ satisfies 
$M_{2,p} = M_{1,p}$ and $G_{2,p} = G_{1,p}$ for $p \in \bseq{0}{n}$ and hence $\fk_2 = \fk_1$ completing the proof.
\end{proof}

\begin{proposition}[Repeated adapted knotset equivalence] \label{prop:repadapteq}
 Let $t \in \bseqo{2}{\infty}$ and suppose $\kn^{\diamond}_s$ is the adapted knotset for $\fk_s$ for $s \in \{1,t\}$ and $\kn_s$ for $s \in \bseq{1}{t-1}$ are knotsets such that $\fk_t = \kn_{t-1} \ast \cdots \ast \kn_1 \ast \fk_1$. Then there exists knotsets $\mathcal{J}_s$ for $s \in \bseq{2}{t}$ such that
\begin{equation*}
    \kn_t^\diamond \ast \fk_t = \kn_t^\diamond \ast\kn_{t-1} \ast \cdots \ast \kn_1 \ast \fk_1=\mathcal{J}_{t} \ast \cdots \ast \mathcal{J}_{2} \ast \kn^\diamond_1 \ast \fk_1.
\end{equation*}
\end{proposition}
\begin{proof}
   The proof follows from repeated applications of Proposition~\ref{prop:adapteq}.
\end{proof}

\begin{proposition}[Terminal knot kernel equivalence]\label{prop:terminalkerneleq}
    Consider model $\fk = (M_{0:n},G_{0:n})$ where $M_n = P_1 \otimes P_2$ and $G_n = H \otimes \phi^{-1}$, and terminal knot $\kn = (n,R,K)$ and let $\kn \ast \fk = (M_{0:n}^\ast,G_{0:n}^\ast)$. The terminal kernels and potential functions satisfy
    \begin{equation*}
        M_n^\ast(G_n^\ast \cdot [1\otimes \varphi]) = M_n(G_n \cdot [1\otimes \varphi]).
    \end{equation*}
\end{proposition}
\begin{proof}
We have,
    \begin{align*}
        M_n^\ast(G_n^\ast \cdot [1\otimes \varphi]) =& (R\otimes K^{H} P_2)([K(H) \otimes \phi^{-1}] \cdot [1\otimes \varphi])  \\
        =& (R\otimes K^{H} P_2)([K(H) \otimes [\phi^{-1} \cdot\varphi])  \\
        =& R\{K(H)\cdot K^{H} P_2(\phi^{-1} \cdot \varphi) \} \\
        =& R\{ K [H\cdot P_2(\phi^{-1} \cdot \varphi) ] \},
    \end{align*}
where the last equality follows from Proposition~\ref{prop:untwist}. Finally, using Definition~\ref{def:terminalkncompat} compatibility condition~(ii) we have
\begin{align*}
        M_n^\ast(G_n^\ast \cdot [1\otimes \varphi])
        =& RK\{H\cdot P_2(\phi^{-1} \cdot \varphi) \} \\
        &= (P_1 \otimes P_2)(H\otimes[\phi^{-1} \cdot \varphi])  \\ 
        &= (P_1 \otimes P_2)(H\otimes[\phi^{-1} \cdot \varphi])\\
        &= M_n(G_n \cdot [1\otimes \varphi]).
\end{align*}
\end{proof}

\begin{center}
{\large\bf SUPPLEMENTARY MATERIAL}
\end{center}

Code to reproduce the experiments and visualisations in this paper is available online: \url{https://github.com/bonStats/KnotsNonLinearSSM.jl}. We gladly acknowledge the \texttt{tidyverse} \citep{tidyverse}, \texttt{matrixStats} \citep{matrixstats}, and \texttt{patchwork} \citep{patchwork} packages in programming language \texttt{R} \citep{Rstats}. As well as the \texttt{Julia} language \citep{Julia-2017} and packages \texttt{DataFrames.jl} \citep{JSSv107i04} and \texttt{Distributions.jl} \citep{JSSv098i16,Distributions.jl-2019}.

\end{document}

%% file: knot-diagram.tex
\begin{tikzpicture}
\tikzset{
  .../.tip={[sep=0pt 1]
    Round Cap[]. Circle[length=0pt 1] Circle[length=0pt 1] Circle[length=0pt 1, sep=0pt]}
    }F
\tikzset{node distance = 1.5cm and 1.5cm}
    
    \draw (0,0) node (Gl1) {$G_t$};
    \node (Kl1) [left=of Gl1] {$M_t$};
    \node (M1l1) [right=of Gl1] {$M_{t+1}$};

    \node (Kl2) [below=of Kl1] {$R_{t}$};
    \node (Gghost) [below=of Gl1] {\phantom{$G_t$}};
    \node (Gl2) [left=of Gghost, yshift=7mm, xshift=13mm] {$G_{t}$};
    \node (Ul2) [right=of Gghost, yshift=7mm, xshift=-13mm] {$K_{t}$};
    \node (M1l2) [below=of M1l1] {$M_{t+1}$};

    \node (Kl3) [below=of Kl2] {$R_{t}$};
    \node (Gl3) [right=of Kl3] {$G_{t}^\prime$};
    \node (M1l3) [right=of Gl3] {$K_{t}^{G_t} M_{t+1}$};

    \draw[] (Kl1.south) -- (M1l1.south);
    \draw[transform canvas={xshift=-10mm}, dashed] (Kl1.south) -- +(0:-10mm);
    \draw[] (Kl1.south) -- +(0:-10mm);
    \draw[transform canvas={xshift=10mm}, dashed] (M1l1.south) -- +(0:10mm);
    \draw[] (M1l1.south) -- +(0:10mm);

    \draw[transform canvas={xshift=-10mm}, dashed] (Kl2.south) -- +(0:-10mm);
    \draw[] (Kl2.south) -- +(0:-10mm);
    \draw[transform canvas={xshift=10mm}, dashed] (M1l2.south) -- +(0:10mm);
    \draw[] (M1l2.south) -- +(0:10mm);

    \begin{knot} [clip width=4, background color=black!5!white]
      \strand (Kl2.south) to [out=right, in=right, looseness=1.5] ([xshift=0mm]$(Gl2)!0.5!(Ul2)$);
      \strand ([xshift=0mm]$(Gl2)!0.5!(Ul2)$) to [out=left, in=left, looseness=1.5] (M1l2.south);
    \end{knot}

    \draw[] (Kl3.south) -- (M1l3.south);
    \draw[transform canvas={xshift=-10mm}, dashed] (Kl3.south) -- +(0:-10mm);
    \draw[] (Kl3.south) -- +(0:-10mm);
    \draw[transform canvas={xshift=10mm}, dashed] (M1l3.south) -- +(0:10mm);
    \draw[] (M1l3.south) -- +(0:10mm);

    \node (FK) [left=of Kl1,yshift=-3mm,xshift=-4mm] {$\fk$};
    \node (FKghost) [below=of FK] {\phantom{$\fk$}};
    \node (KFK) [below=of FKghost,yshift=-2mm] {$\kn\ast\fk$};

    \draw[-Latex, line width=0.5mm] (FK.south) -- (KFK.north);

    \node (d0) [above=of Gl1, yshift=-12mm] {$~$};
    \node (d1) [below=of Gl3, yshift=10mm] {$~$};

\end{tikzpicture}